\documentclass{article}

\usepackage[inner=3.2cm,outer=3.2cm,top=3.2cm,bottom=3.2cm]{geometry} 
\usepackage[utf8]{inputenc}
\usepackage{amsmath} 
\usepackage{amsthm} 
\usepackage{amssymb} 
\usepackage{mathtools} 
\usepackage{tikz} 
\usepackage{graphicx} 
\usepackage{tikz-cd} 
    \usetikzlibrary{cd}
\usepackage{mathrsfs} 
\usepackage{array} 
\usepackage{booktabs} 
\usepackage{setspace} 
\usepackage{xcolor} 
\usepackage{harvard} 
\usepackage{enumitem} 
\usepackage{tcolorbox} 
\usepackage{float}
\usepackage{pdfpages} 
\usepackage{hyperref} 
\usepackage{multicol} 

\hypersetup{
    colorlinks=true,
    linkcolor={blue!50!black},
    citecolor={blue!50!black},
    urlcolor={blue!50!black} 
}

\theoremstyle{definition}
\newtheorem{definition}{Definition}
\theoremstyle{plain}
\newtheorem{lemma}{Lemma}
\theoremstyle{plain}
\newtheorem{theorem}{Theorem}
\theoremstyle{plain}
\newtheorem{proposition}{Proposition}
\theoremstyle{definition}
\newtheorem{example}{Example}
\theoremstyle{definition}

\theoremstyle{plain}
\newtheorem{corollary}{Corollary}

\newcommand{\alloc}{\mathcal{A}}

\newcommand{\pp}{\mathscr{P}}
\newcommand{\obj}{\mathcal{O}}

\DeclareMathOperator*{\argmax}{arg\,max}

\title{Local Priority Mechanisms\thanks{The authors are grateful to Brian Fu for expert research assistance.}}

\author{Joseph Root\thanks{Department of Economics, University of Chicago, Email: jroot@uchicago.edu} \and David S.\ Ahn\thanks{Olin Business School, Washington University in St.\ Louis. Email: ahnd@wustl.edu}}

\date{\date{\today}}

\begin{document}

\maketitle

\begin{abstract}
We introduce a novel family of mechanisms for constrained allocation problems which we call local priority mechanisms. These mechanisms are parameterized by a function which assigns a set of agents -- the local compromisers -- to every infeasible allocation. The mechanism then greedily attempts to match agents with their top choices. Whenever it reaches an infeasible allocation, the local compromisers move to their next favorite alternative. Local priority mechanisms exist for any constraint so this provides a method of constructing new designs for any constrained allocation problem. We give axioms which characterize local priority mechanisms. Since constrained allocation includes many canonical problems as special constraints, we apply this characterization to show that several well-known mechanisms, including deferred acceptance for school choice, top trading cycles for house allocation, and serial dictatorship can be understood as instances of local priority mechanisms. Other mechanisms, including the Boston mechanism, are not local priority mechanisms. We give sufficient conditions for a local priority mechanism to be group strategy-proof. We also provide conditions which enable welfare comparisons across local priority mechanisms.
\end{abstract}

We introduce a new family of mechanisms for constrained allocation problems with private values. Like many mechanisms used in real-world market design, our class is defined by an algorithm. Our algorithm moves through the space of infeasible allocations to approach a final feasible outcome. It begins by trying to give each agent her favorite choice. If that is infeasible, a collection of agents, who we call the local compromisers, then move to their next-favorite choices. The mechanism checks if this adjustment makes the new allocation feasible. If not, then another set of compromisers, defined by the new allocation, must move down their preference. This process continues until the algorithm terminates with a feasible allocation. Constrained allocation encompass several well-known problems \cite{RoAh23}, including social choice, house allocation, one-sided and two-sided matching, and school choice, and our new class of mechanisms provides a unified parameterization of several canonical mechanisms. Since the constraint is left general, it also provides methods for constructing mechanisms given any constrained allocation problem with private values.

The central parameter in our model is the designation of local compromisers to each infeasible allocation. This is akin to the ordering of dictators in serial dictatorship, the priorities in deferred acceptance, or the initial property rights in top trading cycles.  As in these mechanisms, the parameter allows a transparent view into the structure of the mechanism. The assignment of local compromisers can also be viewed as an assignment of local property rights. In every potential conflict of interest, a subset of agents - the local compromisers - are required to relinquish their claims. 

We give three main results. The first describes which kinds of mechanisms can be written as a local priority mechanism for some assignment of compromisers. This characterization unearths what implicit limitations are assumed by restricting attention to local priority mechanisms. It also is methodologically useful in verifying whether known mechanisms are in the class. The second result gives conditions under which a given local priority mechanism is group strategy-proof. This is done through conditions on the paths the algorithm can take while searching through the infeasible space. This provides an understanding of what kinds of assignments are immune to group collusion in misreporting preferences to the planner. This complements an earlier implicit characterization for all mechanisms, local priority or otherwise, provided in \citeasnoun{RoAh23}. Here, since we are considering a specific class of mechanisms, we can provide more explicit conditions. Finally, we provide comparative statics results which enable comparisons across mechanisms and constraints.

The introduction of a new class of mechanisms expands the applicability of market design to new problems. The class of local priority mechanisms provides a recipe for constructing mechanisms for arbitrary constraints. To our knowledge, this is the first general blueprint for defining mechanisms that can be applied to any constrained allocation problem. This is important because for many problems, even those that seem close to canonical ones, there is no off-the-shelf method for constructing mechanisms beyond serial dictatorship. For example, consider a small variation in the house allocation problem where one of the houses can either be assigned to a single agent, or all the agents at the same time. We see no obvious way to adapt top trading cycles to handle this perturbation. On the other hand, we show that local priority mechanisms immediately yield a class of interesting mechanisms for this perturbed problem. Beyond the ability to construct new mechanisms, the fact that an algorithm is attached has additional practical benefits. First, it provides a way of explaining to the participant how the mechanism reaches its final allocation. Second and related, it provides a simple way of implicitly coding the mechanism as a computer program, without having to explicitly list the outcome for every possible submission.

Local priority mechanisms are also theoretically interesting. As mentioned, \citeasnoun{RoAh23} represented several important market design problems as special constraints in a constrained object allocation problem with private values. Local priority mechanisms turn out to include several canonical mechanisms across these environments. We show the following prominent mechanisms are local priority mechanisms: deferred acceptance for the school choice problem; top trading cycles for the house allocation problem; and serial dictatorship generally. Our class of local priority mechanisms unifies these mechanisms and exposes a commonality between them. Each can be implemented by a greedy algorithm that gives priority in a specific way: through school priorities, through property rights, or through the ordering of dictators. The parameter of local priority mechanisms provides a way of understanding these different distributions of power in a unified way. 

This unification does not come from excessive generality. Many mechanisms are not local priority mechanisms. Perhaps the most well-known market design mechanism is the deferred acceptance mechanisms for the marriage problem, where the accepting side has true preferences over the proposers, as opposed to the priorities of schools over students in the school choice problem. We show that the deferred acceptance algorithm for the marriage market is not a local priority mechanism. Another famous mechanism that is not in our class is the Boston mechanism, or immediate acceptance mechanism, for school choice.

The rest of the paper is organized as follows. In section \ref{sec: preliminaries} we introduce the constrained allocation problem, give examples and fix notation for the rest of the paper. In section \ref{sec:constraint_traversing_section} we introduce and characterize local priority mechanisms, giving several examples and non-examples. In section \ref{sec: characterization} we consider incentives and provide two conditions which are sufficient for a local priority mechanism to be group strategy-proof for any constraint. In section \ref{sec: comparisons} we show how to compare welfare across different local priority mechanisms. Finally, in section \ref{sec: discussion} we close with a discussion and outline several open problems. We leave most proofs to the appendix. 

\section{Preliminaries}\label{sec: preliminaries}

Let $N$ be a set of $n$ agents, with generic representatives $i,j,k$. Let $\obj$ denote the set of objects, with typical elements $a,b,c$. Each agent is assigned an individual object from $\obj$ and the family of allocations is $\alloc = \obj^n$, where each coordinate is an agent's allocation. An allocation can also be viewed as a function $\sigma$ from $N$ to $\obj$. Social allocations will typically be given names $x,y,z$ or $\mu,\nu$ when it is convenient view an allocation as a function. For a fixed allocation $\mu$, a suballocation for a subset $M \subset N$ is the restriction $\nu = \mu|_{M}$.

Not all allocations will be feasible. A nonempty constraint $C \subset \alloc$ records the feasible allocations.\footnote{Given this level of generality, there is no loss in having a single set of objects $\obj$ rather than a set of objects $\obj_i$ for each $i$ since in the latter case we can set $\obj=\cup \obj_i$ and design the constraint to require that each agent $i$
s allocation is in $\obj_i$.} Let $\bar{C} = \obj^n - C$ denote the infeasible allocations. For a subset $M \subset N$ of agents, the set of feasible suballocations is defined by:
\begin{align*}
    C^M = \left\{ \, y \in \obj^M \, \, : \, \, \exists \, x \in C \text{  such that } y_j = x_j \, \forall j \in M \, \right\}. 
\end{align*}
Given a subset $B\subset C$ we will also write $B^M = \left\{ \, y \in \obj^M \, \, : \, \, \exists \, x \in B \text{  such that } y_j = x_j \, \forall j \in M \, \right\}$. Let $\nu : M \rightarrow \obj$ be a suballocation. An extension of $\nu$ to a larger subset of agents $M' \supset M$ is a suballocation $\mu : M' \rightarrow \obj$ such that $\mu\vert_{M}=\nu$. The extension $\mu$ is feasible if $\mu \in C^{M'}$. A complete extension is an extension to the set $N$ of all agents. For a fixed suballocation $\nu$, let $C(\nu)$ denote the set of all feasible and complete extensions of $\nu$.

Let $P$ be the set of all strict preferences on $\obj$. We consider a private-values environment where each agent $i$ only cares about her own assigned object $\mu(i)$, and experiences no externalities from others' assignments. Each agent $i$ has a ranking $\succ_i \, \in P$ over the objects. We write $a\succ_i b$ to mean that $i$ ranks object $a$ over object $b$ and $a\succsim_i b$ to mean that either $a\succ_i b$ or $a=b$. Extending the ranking to a partial order over sets, say that $A \succ_i B$ if $a \succ_i b$ for all $a \in A, b \in B$. Let $\pp = P^n$ be the set of preference profiles. Typical profiles will be denoted $\succ$ so that individual preferences will have an index, for example $\succ_i$. As with allocations, we can consider a projection of preference profile to a subset of agents $M \subset N$, denoted by $\succ_M$. If $\succ_M$ is a vector of preferences for agents in $M$ and $\succ_{-M}$ is a vector of preferences for agents not in $M$, then $(\succ_M, \succ_{-M})$ carries its obvious interpretation. We sometimes abuse notation and write $\succ_{-}$ instead of $\succ_{-M}$ when the set $-M$ is clear. It is sometimes useful to record the set of preferred and dis-preferred objects to a given object $a$. Fix a preference $\succ_i$ and an object $a \in \obj$. The  (strict) lower and upper contour sets are 
\[LC_{\succ_i} (a) = \{b \in \obj : a \succ_i b \}\hspace{0.5cm} UC_{\succ_i} (a) = \{b \in \obj : b \succ_i a \}\] respectively. Since we consider strict rankings, it is unambiguous to refer to the $n$th-top choice given a preference $\succ_i$, which we denote $\tau_n (\succ_i)$. Similarly, for a preference profile $\succ$, let $\tau_n (\succ)$ denote the allocation where every agents gets her $n$th choice. Given an allocation $\mu$, let $P^\mu$ be the set of all preference profiles such that each agent $i$ top-ranks $\mu(i)$.

A mechanism must select an outcome from the set $C$ of feasible allocations. A feasible mechanism is a function $f: \pp \rightarrow C$ taking every profile of preferences $\succ$ to a corresponding feasible allocation $f(\succ)$. 

Constrained object allocation captures a variety of important problems as special cases. 

\begin{example}
    The following are prominent examples of constrained allocation problems:
    \begin{enumerate}
        \item Arrovian social choice: $C=\{\mu: \mu(i) = \mu(j) \text{ for all }i,j\}$
        \item House allocation: $C=\{\mu: \mu(i) \neq \mu(j) \text{ for }i\neq j\}$
        \item School choice: for each $a$ there is a capacity $q_a$ so that $C=\{\mu: \vert \mu^{-1}(a)\vert \leq q_a \text{ for all }a\}$
        \item One-sided matching: $\obj=N$ and $C=\{\mu: \mu(\mu(i))=i \text{ for all }i\}$
        \item Two-sided matching: $\obj=N=M\sqcup W$ and $C=\{\mu : \mu(i)\in W\cup\{i\} \text{ for all }i\in M,  \mu(j)\in M\cup\{j\} \text{ for all }j\in W, \text{ and } \mu(\mu(k))=k\text{ for all }k\in N\}$.\footnote{The symbol $\sqcup$ refers to the disjoint union.}
    \end{enumerate}
\end{example}

The classical social choice problem can be represented as a constraint by requiring that all agents be assigned the same object. By contrast, house allocation requires that no two agents be assigned the same object. Of course, for the set of feasible allocations to be nonempty, one must have that $\vert \obj \vert \geq n$.  In school choice, each school has a capacity and the only constraint is that the allocation not exceed the capacity of any school. Here, the feasible set is nonempty if and only if $\sum_a q_a \geq n$. Note that house allocation is the special case of school choice where each capacity is $1$. In one-sided matching, also known as the roommates problem, the objects are potential partners. The constraint is that if an agent $i$ is matched with an agent $j$, then $j$ must also be matched with $i$. Agents can also be unmatched, meaning $\mu(i)=i$. These two possibilities are captured by the requirement that $\mu(\mu(i))=i$. In two-sided matching, the set of agents is partitioned into two groups $M$ and $W$ and there is the additional requirement that if $i$ and $j\neq i$ are matched, they are matched across the two sides.

\section{Local Priority Mechanisms}\label{sec:constraint_traversing_section}

We now introduce the novel class of mechanisms that is the main contribution of the paper. These mechanisms are implemented by ``greedy'' algorithms that search through the set of allocations in steps, approaching the feasible set in sequence. At each step, a new allocation is considered and if the algorithm has not terminated, it is because the current allocation is infeasible. Some set of agents must be assigned different objects to achieve feasibility. The algorithm begins with the allocation where each agent receives her favorite object. If this is feasible, then it ends. If this is not feasible, a set of agents must compromise. Those agents are then assigned their next-favorite objects. The algorithm checks whether the new allocation is now feasible. If not, then another set of agents are required to compromise. 

The key parameter of the algorithm is the set of compromisers assigned to each infeasible allocation. This assignment is fixed and is independent of preferences.

\begin{definition}
    A \textbf{local compromiser assignment} is a map $\alpha:\alloc \rightarrow 2^{N}$ such that for every infeasible $x\in \bar{C}, \alpha(x)$ is nonempty and for every feasible $y \in C$, $\alpha(y)=\emptyset$. 
\end{definition}
An agent $i\in \alpha(x)$ is referred to as a \textbf{local compromiser} at $x$. This assignment parameterizes the following algorithm, called the \textbf{local priority algorithm} for $\alpha$, which takes a profile of preferences as an input and returns a feasible allocation, or, if unable to do so, returns the symbol $\emptyset$. For a given preference profile $\succ$:
\begin{tcolorbox}
    \fbox{Step 0} Let $x^{0}=\tau_{1}(\succ)$ 
    
    \vspace{0.2cm}
    \fbox{Step $t\geq 1$} If $x^{t-1}$ is feasible, stop and return $x^{t-1}$. Otherwise, if there is any $i\in \alpha(x^{t-1})$, such that $LC_{\succ_{i}}(x^{t-1}_{i})$ is empty, stop and return $\emptyset$. If not, define $x^{t}_{j}=x^{t-1}_{j}$ for all $j \notin \alpha(x^{t-1})$ and for all $k\in \alpha(x^{t-1})$ let $x^{t}_{k}=\argmax_{\succ_{k}}LC_{\succ_{k}}(x^{t-1}_{k})$.
\end{tcolorbox}

If the algorithm terminates in $m$ steps, we call $x^0,\dots x^{m-1}$ the allocations considered by the local priority algorithm under $\alpha$. Given two sequences of allocations $x^0,\dots x^m$ and $y^0,\dots, y^p$ we say that the second is a truncation of the first if there is an $l$ such that $m=l+p$ and $y^0=x^l$, $y^1=x^{l+1}$ and so on until $y^p=x^{l+p}=x^{m}$. There is no guarantee that the local priority algorithm will terminate in an allocation as it might exhaust an agents' possible allocations. When the local priority algorithm yields a well-defined allocation for every preference profile, that is, if the algorithm never terminates in the symbol $\emptyset$, we say that $\alpha$ is \textbf{implementable}. In this case, the algorithm determines a well-defined mechanism which we call the \textbf{local priority mechanism under} $\alpha$ and denote as $LP_{\alpha}$. More generally, a \textbf{local priority mechanism} is a feasible mechanism $f:\pp \rightarrow C$ such that there is some local compromiser assignment $\alpha$ such that for any preference profile $\succ$, $LP_{\alpha}(\succ)=f(\succ)$. In this case, $\alpha$ is said to induce $f$.

This construction can be used both to generate new mechanisms and to analyze existing mechanisms. We now characterize the class of mechanisms which are local priority mechanisms. That is, we describe local priority mechanisms as a subset of all mechanisms. This exercise is analogous to finding which mechanisms for the house allocation problem are equivalent to top trading cycles for some assignment of property rights; while is is straightforward to construct the resulting mechanism from a fixed assignment of property rights, it leaves open whether an abstract social choice function is a top-trading-cycles mechanism. This provides some transparency to what types of rules are excluded by limiting attention to this class of mechanisms. It also provides a convenient test of whether a general mechanism is a local compromiser mechanism. This test will later prove useful for understanding whether several canonical mechanisms are local priority mechanisms.

Three conditions characterize local priority mechanisms. First, all local priority mechanisms are unanimous; if all agents top-rank an allocation which is feasible, it is chosen. The second condition is related to the familiar tops-only property from the social choice literature.\footnote{See for instance \citeasnoun{weymark2008strategy} and \citeasnoun{chatterji2011tops}} Tops-only requires that the outcome of a mechanism depends only on the top choices of the agents. More formally, recall that $P^{\mu}$ is the set of all preference profiles where the vector of favorite objects across agents is the allocation $\mu$. The tops-only property requires that the mechanism be constant on $P^{\mu}$. Local priority mechanisms are generally not tops-only. However, there is a specific invariance across $P^{\mu}$. For any infeasible $\mu$, there must be some agent who never gets her top choice across all of $P^{\mu}$ irrespective of the rest of the preference profile below the top choices. The set of such agents is identified with the local compromisers.

The final condition requires that the mechanism is unresponsive to a compromiser's ranking in a particular way. If the compromiser's favorite choice is moved to the bottom of her ranking, while other preferences are held constant, then the outcome of the mechanism does not change.
This condition is similar in spirit to Maskin monotonicity but is much weaker. Like our condition, Maskin monotonicity requires a particular invariance across profiles. Suppose that $f(\succ)=\mu$ and $\succ'$ is a profile such that for any agent $i$ and any object $b$ if $\mu(i)\succ_i b$ then $\mu(i)\succ_i' b$. Maskin monotonicity requires that $f(\succ')=\mu$. The well-known intuition is that moving an option up in a preference ranking cannot hurt its status as the social choice. Compromiser invariance is clearly implied by Maskin monotonicity. In Appendix \ref{appendix: preliminary} we show that Maskin monotonicity is equivalent to group strategy-proofness in our setting. In Section \ref{sec: DA} we show that the deferred acceptance algorithm is a local priority mechanism, though it is well-known to violate group strategy-proofness. Thus compromiser invariance is strictly weaker than Maskin monotonicity.

These conditions characterize the class of local priority mechanisms as shown in the following proposition.

\begin{proposition} \label{nec suff conditions}
    Fix a constraint $C$. A feasible mechanism $f$ is a local priority mechanism if and only if it satisfies:
    \begin{enumerate}
        \item (Unanimity) If $\tau_1(\succ)$ is feasible, then $f(\succ) = \tau_1(\succ)$.
    
        \item (Fixed compromiser) For all infeasible allocations $\mu$, 
        \[ \bigcap_{\succ \in P^{\mu}} \{ i : f_i (\succ) \neq \mu_i \} \neq \emptyset.\]
        
        \item (Compromiser invariance) For all $\mu$ and $\succ \, \in P^\mu$, if $M$ is the set of agents such that $f(\succ)(i) \neq \mu(i)$ for all $\succ \, \in P^\mu$, then for the preference profile $\succ'$ in which all agents $i\in M$ bottom-rank $\mu(i)$ and rank every other object the same as in $\succ_i$, we have $f(\succ') = f(\succ)$.
    \end{enumerate}
\end{proposition}

As mentioned, compromiser invariance is implied by Maskin monotonicity. However, Maskin monotonicity is equivalent to group strategy-proofness, yielding the following result.

\begin{corollary}\label{cor: gsp}
    A mechanism $f$ is a group strategy-proof local priority mechanism if and only if it satisfies Unanimity, Fixed compromiser, and Maskin monotonicity.
\end{corollary}

\subsection{Examples}

We now show that many well-known mechanisms are local priority mechanisms and provide two examples of mechanisms which are not local priority. We find that deferred acceptance for school choice, top trading cycles for house allocation, and serial dictatorship are all special cases of local priority mechanisms. However, the Boston mechanism for school choice and deferred acceptance for the marriage problem are not local priority mechanisms.

We first note that serial dictatorship is a local priority mechanism for any constraint. This serves to show that the set of local priority mechanisms is nonempty for every constraint. Specifically, fix an ordering of the agents, $\sigma:\{1,2,\dots, n\}\rightarrow N$. Let agent $\sigma(1)$ pick their favorite object from $C^{\{\sigma(1)\}}$. Let $\nu_1$ be this suballocation. Now let agent $\sigma(2)$ pick their favorite object from $C(\nu_1)^{\{\sigma(2)\}}$, the set of objects for $\sigma(2)$ compatible with agent $1$'s selection. Let $\nu_2$ be this suballocation for agents $1$ and $2$. Continue in this way until all agents have been allocated an object. This mechanism can be realized as a local priority mechanism for any constraint $C$ by setting $\alpha(\mu)=\{\sigma(r)\}$ where $r$ 
is the smallest integer such that $\mu(\sigma(r))\notin C(\mu\vert_{\sigma(1),\dots,\sigma(r-1)})^{\{\sigma(r)\}}$ for any infeasible $\mu$. That is, we simply set $\alpha(\mu)$ to be the first dictator whose allocation is not compatible with the assignments of earlier dictators. Unanimity and compromiser invariance are immediate.

\subsubsection{Deferred Acceptance}\label{sec: DA}

In this section, we demonstrate that the deferred acceptance mechanism applied to the school choice problem is a local priority mechanism. What makes the school choice problem different than, say, the two-sided marriage problem is that the priorities of the schools are interpreted as parameters of the mechanism, rather than preferences of agents. The only preferences that are important for welfare are those on the students' side. In contrast, the marriage problem involves preferences across both sides of the markets. This distinction will turn out to be important. Here, we demonstrate that deferred acceptance is a local priority mechanism for the school choice setting. Later, we will show that, in contrast, the deferred acceptance mechanism applied to the marriage problem is not a local priority mechanism. 

 The school choice problem consists of a set of students $N$ and a set of schools $\obj$. Each student $i$ has a preference $\succ_i$ over schools and each school $s$ has a capacity $q_s$ reflecting their maximum enrollment. We assume that $\sum q_s \geq \vert N \vert$, since a dummy school with capacity $N$ can always be added as an outside option. Schools also have strict priority rankings $>_s$ over the students.\footnote{The priority of the outside option does not matter; it will never reject a student.} We now introduce the deferred acceptance mechanism \cite{GaSh62}. The mechanism can be implemented using a number of related algorithms. Here we introduce the cumulative deferred acceptance algorithm since this implementation makes the connection with local priority mechanisms most evident.
     
\begin{definition}
    The \textbf{cumulative deferred acceptance algorithm} takes as input a preference profile $\succ$ and returns an allocation $DA(\succ)$. The process is described as follows:

    \begin{tcolorbox}
        \fbox{Step $t\geq 1$} All students apply to their favorite school who has not yet rejected them. Any student $i$ who applied to a school $s$ such that there are $q_s$ students with higher priority at $s$ who also applied this round is rejected. If there are no rejections, return the allocation in which all students are assigned to the school they applied to this round. Otherwise, move to the next round. 
    \end{tcolorbox}
    The resulting mechanism is called the \textbf{deferred acceptance mechanism}.
\end{definition}

To view this as a local priority mechanism, let $CDA^t(\succ)$ be the allocation where every student is assigned the school where they applied in round $t$. If the DA algorithm terminates in $m$ steps, this gives a sequence of allocations $CDA^1(\succ), CDA^2(\succ), \dots , CDA^m(\succ)$ where $CDA^t(\succ)$ is infeasible for $t<m$ and $CDA^m(\succ)$ is the feasible allocation returned by the deferred acceptance algorithm. For any allocation $x$, let $\alpha(x)$ be the set of agents $i$ such that $\vert\{j : x_j=x_i \text{ and }j>_{x_i} i \}\vert \geq q_{x_i}$. By construction, the allocations considered by the local priority mechanism under $\alpha$ is exactly $CDA^1(\succ), CDA^2(\succ), \dots , CDA^m(\succ)$. This proves the following fact.

\begin{proposition}
     Deferred acceptance is a local priority mechanism.
\end{proposition}

We demonstrate this with the following simple example.

\begin{example}
There are three students, $1, 2, 3$, and three schools, $a, b, c$ each with capacity $1$. Let the preferences and priorities be as follows:
\begin{center}
\begin{tabular}{ccccccc}
$\succ_1$ & $\succ_2$ & $\succ_3$ & $\hphantom{\ldots}$ & $>_{a}$ & $>_b$ & $>_c$ \\
\cmidrule(lr){1-1}
\cmidrule(lr){2-2}
\cmidrule(lr){3-3}
\cmidrule(lr){5-5}
\cmidrule(lr){6-6}
\cmidrule(lr){7-7}
$a$ & $a$ & $b$ &  & $3$ & $1$ & $1$\\
$\textcolor{blue}{b}$ & $b$ & $\textcolor{blue}{a}$ &  & $1$ & $2$ & $2$\\
$c$ & $\textcolor{blue}{c}$ & $c$ &  & $2$ & $3$ & $3$\\
\end{tabular}
\end{center}
The deferred acceptance algorithm results in the following steps:
\begin{tcolorbox}
\fbox{Step 1} Students $1$ and $2$ apply to $a$, and 3 applies to $b$. Student $2$ is rejected from $a$. 
\vspace{0.2cm}\\
\fbox{Step 2} Student $1$ applies to $a$, and students $2$ and $3$ apply to $b$. School $b$ rejects $3$.
\vspace{0.2cm}\\
\fbox{Step 3} Students $1$ and $3$ apply to $a$, and students $2$ applies to $b$. School $a$ rejects $1$.
\vspace{0.2cm}\\
\fbox{Step 4} Students $1$ and $2$ apply to $b$, and student $3$ applies to $a$. School $b$ rejects $2$.
\vspace{0.2cm}\\
\fbox{Step 5} Student $1$ applies to $b$, $2$ applies to $c$, and $3$ applies to $a$. There are no rejections. Students are matched to the schools to which they applied this round.
\end{tcolorbox}

Now consider the local priority algorithm induced by the local compromiser assignment shown in the figure below. Each of the $27$ possible allocations are listed. The infeasible allocations are shaded grey. Agent $3$'s allocation is determined by the three panels, so that the top-right allocation is $(a,c,c)$ which is infeasible since $2$ and $3$ can't both be assigned $c$. The allocation on the bottom-left is $(c,a,a)$. The local compromisers are listed in the infeasible squares. For example, in the infeasible square $(a,a,a)$ both agents $1$ and $2$ are local compromisers. 
\begin{figure}[H]\label{fig: DA example}
\centering
    \label{Deferred Acceptance Local Priority Mech}
    \includegraphics[width=0.9\linewidth]{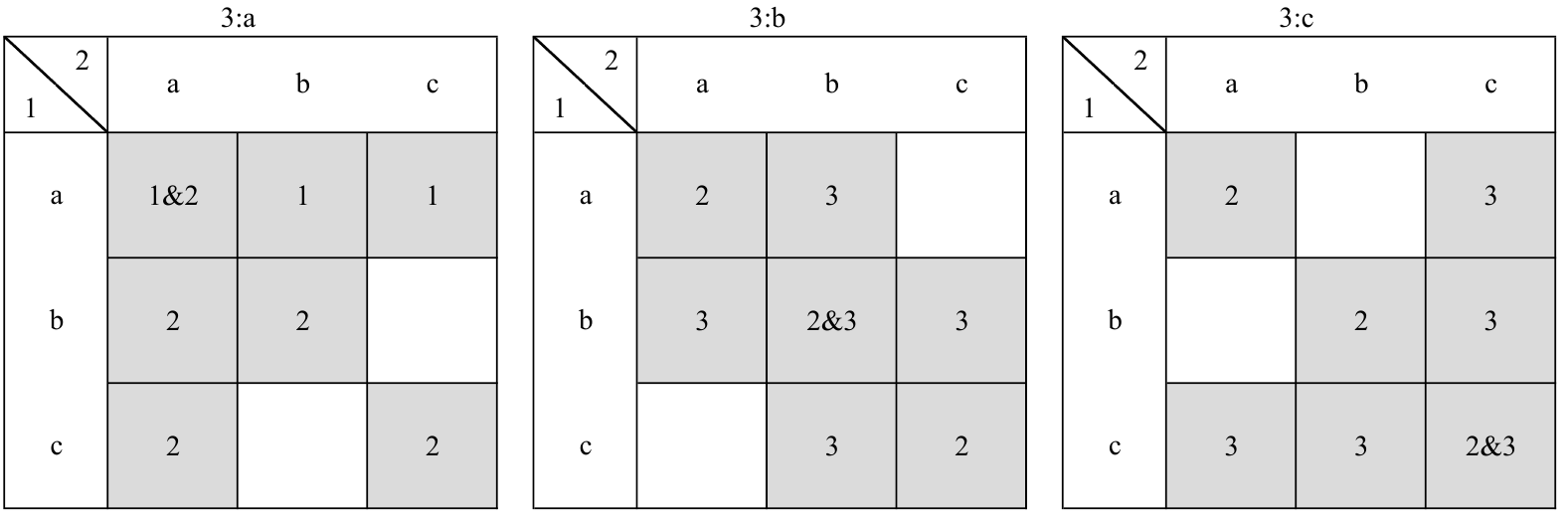}
    \caption{The local priority mechanism for the school priorities above.}
\end{figure}

\noindent The local priority algorithm considers the following allocations in sequence
\[
\begin{pmatrix} a &a &b \end{pmatrix} \to \begin{pmatrix} a &b &b \end{pmatrix} \to\begin{pmatrix} a &b &a \end{pmatrix} \to\begin{pmatrix} b &b &a \end{pmatrix} \to\begin{pmatrix} b &c &a \end{pmatrix} 
\]
\noindent resulting in the same allocation.
\end{example}

\subsubsection{Top Trading Cycles}

The top trading cycles mechanism introduced in \citeasnoun{ShSc74} is a prominent mechanism for the house allocation problem where $\vert \obj \vert = \vert N \vert$. We describe the mechanism below and show that it can be recast as a local priority mechanism. Since \citeasnoun{ShSc74}, many generalizations have been introduced. Two notable examples are the hierarchical exchange mechanisms of \citeasnoun{Papai00} and the broker and braiding mechanisms of \citeasnoun{PyUn17} and \citeasnoun{Bade16}. These are also local priority mechanisms.\footnote{ These mechanisms have already been shown to be efficient group strategy-proof \cite{PyUn17}. Since group strategy-proofness is equivalent to Maskin monotonicity, we immediately get compromiser invariance. Unanimity is implied by efficiency. It remains to show that the fixed compromiser condition holds. This, however, is immediate from the definition of the associated algorithms. In the appendix for this section, we give a more elementary proof of the fact that top-trading cycles is a local priority mechanism. That is, a proof that makes no reference to the nontrivial proof that top trading cycles is group strategy-proof.}

The top-trading cycles mechanism is parameterized by an initial ownership structure described by a bijection $\mu^0:N\rightarrow \obj$

\begin{definition} Given an initial ownership structure $\mu^0$, the \textbf{top-trading cycles algorithm} takes as input a preference profile $\succ$ and returns a feasible allocation $TTC(\succ)$ as follows:

\begin{tcolorbox}
    \fbox{Step $t\geq 1$} Each remaining agent points to the agent endowed with their favorite remaining object. For any cycle, defined as a collection of agents $i_1, \cdots i_m$ such that $i_1$ points to $i_2$, $i_2$ points to $i_3$ and so on while $i_m$ points to $i_1$, these agents trade their endowments along the cycle. These agents leave the algorithm with their new object. If all agents are assigned an object, stop. Otherwise, proceed to the next step.
\end{tcolorbox}
\noindent The resulting mechanism is called the \textbf{top trading cycles mechanism}.\footnote{It is simple to verify that there will be at least one cycle at every step.}
\end{definition}

\begin{proposition}\label{proposition: TTC}
For any initial ownership structure, top trading cycles is a local priority mechanism.
\end{proposition}

We prove this proposition in Appendix \ref{appendix: TTC} using Proposition \ref{nec suff conditions}, however the following simple example demonstrates the connection.

\begin{example}
There are three agents $1,2,3$ and preferences are as follows:
\begin{center}
\begin{tabular}{ccc}
$\succ_1$ & $\succ_2$ & $\succ_3$\\
\cmidrule(lr){1-1}
\cmidrule(lr){2-2}
\cmidrule(lr){3-3}
$a$ & $\textcolor{blue}{b}$ & $\textcolor{blue}{a}$ \\
$\textcolor{blue}{c}$ & $a$ & $c$ \\
$b$ & $c$ & $b$ \\
\end{tabular}
\end{center}
Let $\mu^0(1)=b$, $\mu^0(2)=c$ and $\mu^0(3)=a$. The top-trading cycles algorithm proceeds as follows:
\begin{tcolorbox}
\fbox{Step 1} Agent 1 points to 3, 2 points to 1, 3 points at herself. Agent 3 leaves with $a$.
\vspace{0.2cm}\\
\fbox{Step 2} Agent 1 points to 2 and 2 points to 1. Agents 1 and 2 trade endowments.
\end{tcolorbox}

This results in the allocation $(c,b,a)$. Alternatively, consider the local priority mechanism with local compromiser assignments as shown in the figure below.
\begin{figure}[H]
\centering
    \label{TTC Local Priority Mech}
    \includegraphics[width=0.9\linewidth]{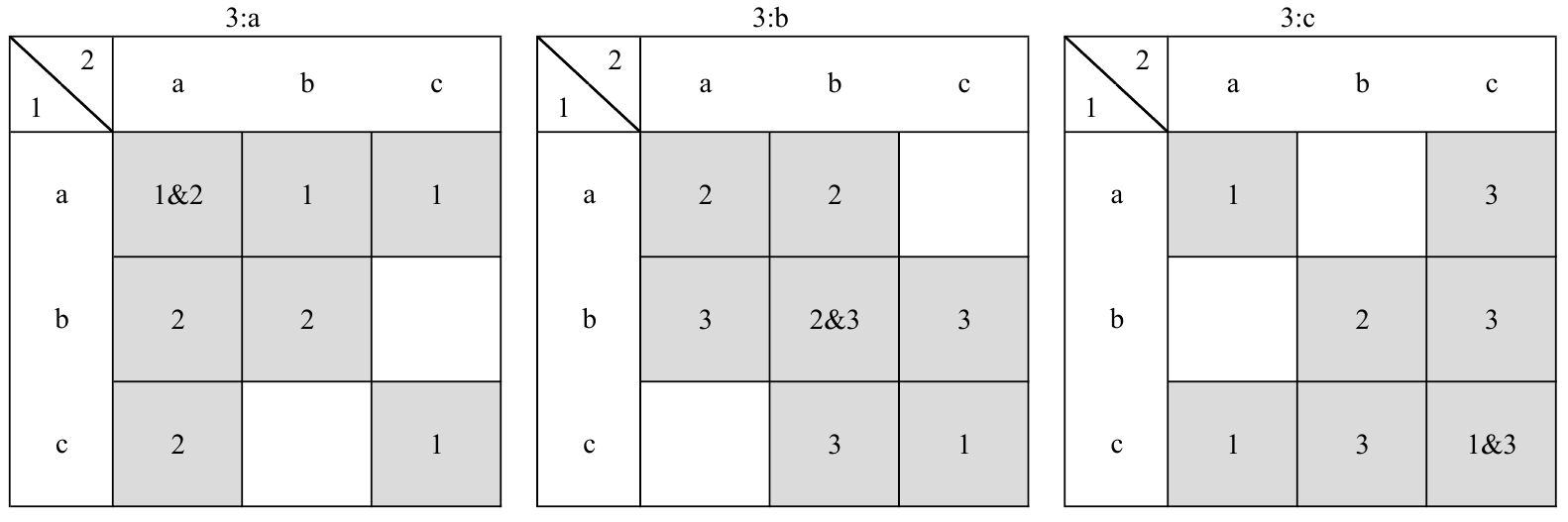}
    \caption{The local priority mechanism for 3 agent top trading cycles.}
\end{figure}

Utilizing this local priority mechanism, we again find the same final allocation of $(c, b, a)$ by moving through the following allocations:
\[
\begin{pmatrix} a &b &a \end{pmatrix} \to \begin{pmatrix} c &b &a \end{pmatrix}.
\]
\end{example}

To see why this particular choice of $\alpha$ works, consider the infeasible allocation $(a,b,a)$. Since $a$ is endowed to agent $3$, $1$ must be the local compromiser at this allocation. For another example, consider the infeasible allocation $(a,a,b)$. Now no one is assigned their endowment. However, since both $1$ and $2$ top-rank $a$ which is owned by $3$ and since $3$ top-ranks $b$, which is owned by $1$, agent $2$ must be the compromiser. More generally, for any $x$, generate the directed graph with one node for each agent and one edge directed away from each node $i$ to the agent $j$ such that $x_i=\mu^0(j)$. Then $\alpha(x)$ is the set of agents, not in a cycle, but who are pointing to an agent in a cycle. 

\subsection{Non-Examples}

While several important mechanisms are local priority mechanisms, this class is far from including all mechanisms. In this section we demonstrate that there are important mechanisms that are not local priority mechanisms. Some well-known examples will be shown here---these serve to demonstrate how the fixed compromiser and compromiser invariance conditions can be violated.

\subsubsection{Immediate Acceptance}

The first example of a mechanism that is not a local priority mechanism is the the immediate acceptance mechanism, also known as the Boston mechanism for school choice. It is well-known to have poor incentive properties which ultimately led to its replacement in many school districts \cite{abdulkadiroglu2006changing}. Nevertheless, some have argued that immediate acceptance can have appealing efficiency properties \cite{abdulkadirouglu2011resolving}. As in the case of deferred acceptance, each school $s$ has a priority ordering $>_s$ and a capacity $q_s$ which parameterize the mechanism. As before, we assume that $\sum_s q_s \geq n$.

\begin{definition}
The \textbf{immediate acceptance algorithm} takes as input a preference profile and returns a feasible allocation $IA(\succ)$ as follows:
\begin{tcolorbox}
    \fbox{Step $t\geq 1$} All unassigned students apply to their $t$-th ranked school. Any school $s$ with no remaining seats rejects all applicants. For any school with vacant seats, the highest priority students are admitted up to the capacity. If all students are assigned a seat, stop. Otherwise proceed to the next round.
\end{tcolorbox}
The resulting mechanism is called the \textbf{immediate acceptance mechanism}.
\end{definition}
Immediate acceptance is not a local priority mechanism, as it does not fulfill compromiser invariance. For instance, let there be $3$ schools $a, b, c$ and $q_a+q_b+1$ agents. Fix an agent $i$ and a profile $\succ$ where (1) $i$'s preference over the schools is $a\succ_i b\succ_i c$, (2) $q_a$ other students top-rank $a$ and (3) $q_b$ students top-rank $b$. Suppose that $a$ prioritizes $i$ last and $b$ prioritizes $i$ first. Then agent $i$ is a fixed compromiser at the profile $\succ$, but if they were to announce $\succ_i'$ with $b\succ_i' c \succ_i' a$, they would be accepted by school $b$ in step $1$, while under $\succ_i$ they would have been rejected since school $b$ is filled in the first step. So, $f(\succ) \neq f(\succ_i', \succ_{-i})$.

\subsubsection{Deferred Acceptance in Marriage Markets}
The deferred acceptance mechanism is now widely used in school choice settings. However, it was originally formulated for one-to-one matching in two-sided markets where each side has preferences over the other \cite{GaSh62}. Unlike school choice, where schools' preferences are a parameter of the mechanism, here the full vector of preferences is the input of the mechanism. For this application, the deferred acceptance algorithm is not a local priority mechanism. 

Consider the following example with $3$ agents on either side and where the $m$ agents are the proposers in the mechanism. Suppose we have the following vector of top-choices. 
\begin{center}
\begin{tabular}{ccccccc}
$\succ_{m_1}$ & $\succ_{m_2}$ & $\succ_{m_3}$ & $\hphantom{\ldots}$ & $\succ_{w_1}$ & $\succ_{w_2}$ & $\succ_{w_3}$ \\
\cmidrule(lr){1-1}
\cmidrule(lr){2-2}
\cmidrule(lr){3-3}
\cmidrule(lr){5-5}
\cmidrule(lr){6-6}
\cmidrule(lr){7-7}
$w_1$ & $w_2$ & $w_2$ &  & $m_3$ & $m_1$ & $m_3$\\
$\cdot$ & $\cdot$ & $\cdot$ &  & $\cdot$ & $\cdot$ & $\cdot$\\
$\cdot$ & $\cdot$ & $\cdot$ &  & $\cdot$ & $\cdot$ & $\cdot$\\
\end{tabular}
\end{center}

\noindent Consider the following three preference profiles. The outcome of the deferred acceptance algorithm is shown in blue for each. Only the rankings required to calculate the outcome are specified.

\begin{center}
\begin{tabular}{ccccccc}
$\succ_{m_1}$ & $\succ_{m_2}$ & $\succ_{m_3}$ & $\hphantom{\ldots}$ & $\succ_{w_1}$ & $\succ_{w_2}$ & $\succ_{w_3}$ \\
\cmidrule(lr){1-1}
\cmidrule(lr){2-2}
\cmidrule(lr){3-3}
\cmidrule(lr){5-5}
\cmidrule(lr){6-6}
\cmidrule(lr){7-7}
$\textcolor{blue}{w_1}$ & $\textcolor{blue}{w_2}$ & $w_2$ &  & $m_3$ & $m_1$ & $\textcolor{blue}{m_3}$\\
$\cdot$ & $\cdot$ & $\textcolor{blue}{w_3}$ &  & $\textcolor{blue}{m_1}$ & $\textcolor{blue}{m_2}$ & $\cdot$\\
$\cdot$ & $\cdot$ & $w_1$ &  & $m_2$ & $m_3$ & $\cdot$\\
\end{tabular}
\end{center}

\begin{center}
\begin{tabular}{ccccccc}
$\succ_{m_1}$ & $\succ_{m_2}$ & $\succ_{m_3}$ & $\hphantom{\ldots}$ & $\succ_{w_1}$ & $\succ_{w_2}$ & $\succ_{w_3}$ \\
\cmidrule(lr){1-1}
\cmidrule(lr){2-2}
\cmidrule(lr){3-3}
\cmidrule(lr){5-5}
\cmidrule(lr){6-6}
\cmidrule(lr){7-7}
$\textcolor{blue}{w_1}$ & $w_2$ & $\textcolor{blue}{w_2}$ &  & $m_3$ & $m_1$ & $m_3$\\
$\cdot$ & $\textcolor{blue}{w_3}$ & $\cdot$ &  & $\textcolor{blue}{m_1}$ & $\textcolor{blue}{m_3}$ & $\textcolor{blue}{m_2}$\\
$\cdot$ & $w_1$ & $\cdot$ &  & $m_2$ & $m_2$ & $m_1$\\
\end{tabular}
\end{center}

\begin{center}
\begin{tabular}{ccccccc}
$\succ_{m_1}$ & $\succ_{m_2}$ & $\succ_{m_3}$ & $\hphantom{\ldots}$ & $\succ_{w_1}$ & $\succ_{w_2}$ & $\succ_{w_3}$ \\
\cmidrule(lr){1-1}
\cmidrule(lr){2-2}
\cmidrule(lr){3-3}
\cmidrule(lr){5-5}
\cmidrule(lr){6-6}
\cmidrule(lr){7-7}
$w_1$ & $w_2$ & $w_2$ &  & $\textcolor{blue}{m_3}$ & $\textcolor{blue}{m_1}$ & $m_3$\\
$\textcolor{blue}{w_2}$ & $\textcolor{blue}{w_3}$ & $\textcolor{blue}{w_1}$ &  & $m_1$ & $m_2$ & $\textcolor{blue}{m_2}$\\
$w_3$ & $w_1$ & $w_3$ &  & $m_2$ & $m_3$ & $m_1$\\
\end{tabular}
\end{center}

\noindent Since each agent is given their top choice in at least one of these outcomes, deferred acceptance cannot satisfy the fixed compromiser condition.

\section{Incentives and Efficiency}\label{sec: characterization}

In this section, we introduce properties of the local compromiser assignment which guarantee that the resulting mechanisms will have robust incentive properties.  

\begin{definition}
A mechanism $f$ is \textbf{group strategy-proof} if, for every $\succ \, \in \pp$ and every $M\subset N$, there is no $\succ'_{M}$ such that
\begin{enumerate}
    \item $f_{j}(\succ'_{M},\succ_{-M})\succsim_{j}f_{j}(\succ) \text{  for all  }j\in M $;
    \item $f_{k}(\succ'_{M},\succ_{-M})\succ_{k}f_{k}(\succ)$ for at least one $k\in M$.
\end{enumerate}
\end{definition}

A mechanism is individually strategy-proof if the above holds for all coalitions of size $1$. It is worth remarking here that while, in principle, group strategy-proofness rules out manipulations by coalitions of any size, it is necessary and sufficient to rule out misreports by groups of size at most two \cite{Alva17}.

Group strategy-proofness is a strong incentive property. It is satisfied by top-trading cycles, but not by deferred acceptance. It is equivalent to either Maskin monotonicity, or to individual strategy-proofness and non-bossiness. Appendix \ref{appendix: preliminary} defines these conditions formally and proves these equivalences. Group strategy-proofness is also appealing since it gives efficiency properties for free. Recall that an allocation $\mu$ is Pareto efficient if there is no other feasible allocation $\nu$ which all agents weakly prefer and some agents strictly prefer. A mechanism is Pareto efficient if for every profile of preferences, it outputs a Pareto efficient allocation. 

\begin{proposition}\label{local priority PE}
If a local priority mechanism is group strategy-proof, it is Pareto efficient. However, a local priority mechanism can be Pareto efficient, but not group strategy-proof. 
\end{proposition}

It will be useful to now address the fact that local priority mechanisms may not have a unique local compromiser assignment that implements it. That is, it is easy to find examples of two different local compromiser assignments $\alpha$ and $\alpha'$ that induce the same mechanism $LP_\alpha=LP_{\alpha'}$. Such an example is given in Appendix \ref{appendix: non-uniqueness example}.  However, the next proposition shows that if the resulting mechanism is group strategy-proof, $LP_\alpha=LP_{\alpha'}=LP_{\alpha\cup \alpha'}$.\footnote{The notation $\alpha\cup \alpha'$ refers to the function which is the pointwise union of $\alpha$ and $\alpha'$.}

\begin{proposition} \label{local compromiser closure}
Let $f$ be a group strategy-proof local priority mechanism and let $A$ be the set of local compromiser assignments which induce $f$. Then $A$ is closed under (pointwise) unions.
\end{proposition}

So while there are potentially many $\alpha$ which induce the same mechanism, the pointwise union of them is a natural representative. Henceforth, when we refer to \textit{the} local compromiser assignment for a given group strategy-proof local priority mechanism we will mean the pointwise union of all local compromiser assignments which induce $f$.

\subsection{Group Strategy-proof Local Priority Mechanisms}

We now give two reasons a local priority mechanism can fail to be group strategy-proof. By ruling out these possibilities, we can generate group strategy-proof mechanisms. 

To help understand the possible opportunities for misreporting, we present concrete examples. First, consider the following: 

\begin{example}
There are three agents, three objects and a set of infeasible allocations together with a local compromiser assignment as shown in the figure below.
    \begin{figure}[H]
        \centering
        \includegraphics[width=0.9\linewidth]{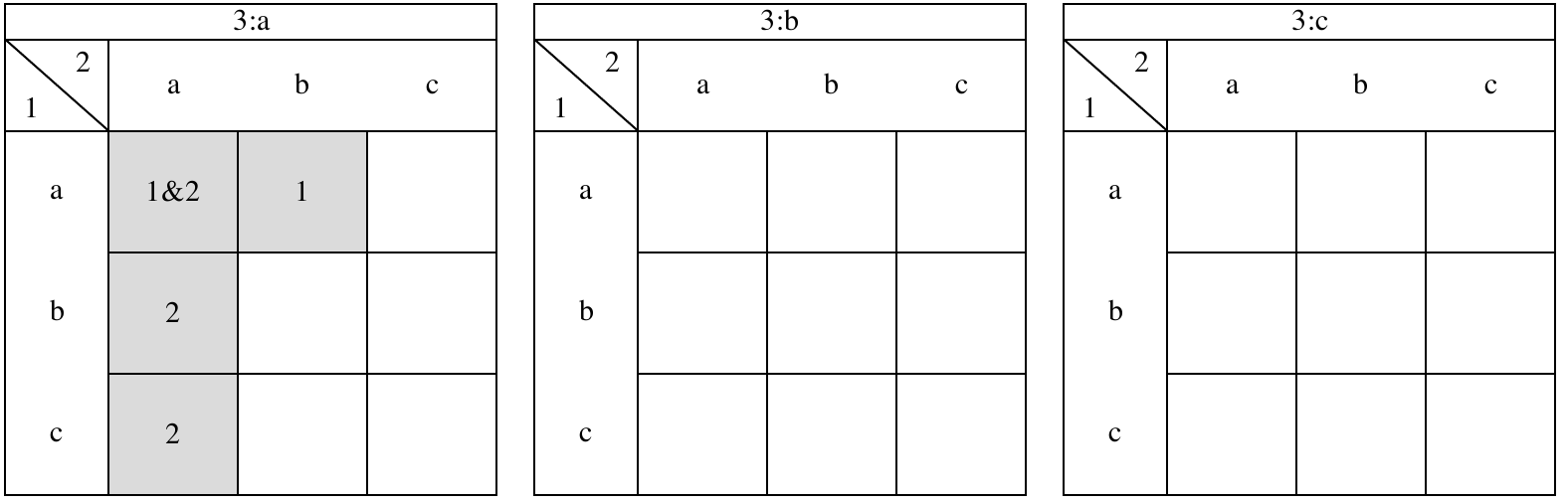}
        \label{fig: GSP violation 1}
    \end{figure}
Now consider the two profiles
   \begin{center}
        \begin{tabular}{ccccccc}
        $\succ_{1}$ & $\succ_{2}$ & $\succ_{3}$ & $\hphantom{\ldots}$ & $\succ_{1}$ & $\succ_{2}'$ & $\succ_{3}$ \\
        \cmidrule(lr){1-1}
        \cmidrule(lr){2-2}
        \cmidrule(lr){3-3}
        \cmidrule(lr){5-5}
        \cmidrule(lr){6-6}
        \cmidrule(lr){7-7}
        $a$ & $a$ & $\textcolor{blue}{a}$ &  & $\textcolor{blue}{a}$ & $\textcolor{blue}{c}$ & $\textcolor{blue}{a}$\\
        $\textcolor{blue}{b}$ & $\textcolor{blue}{c}$ & $b$ &  & $b$ & $b$ & $b$\\
        $c$ & $b$ & $c$ &  & $c$ & $a$ & $c$\\
        \end{tabular}
    \end{center} 
In the first step of the local priority algorithm under the first preference profile, both $1$ and $2$ are required to compromise. Agent $1$ compromises to $b$ and agent $2$ to $c$. The allocation $(b,c,a)$ is feasible so the algorithm returns this allocation. However, agents $1$ and $2$ could form a coalition and jointly misreport. If $1$ reports $\succ_1$ as before and $2$ reports $\succ_2'$ this has no effect on $2$'s allocation but makes $1$ strictly better, violating group strategy-proofness.
\end{example}

The issue in this example is that although both agents $1$ and $2$ are required to compromise at the allocation $(a,a,a)$, $2$ can preempt this by top-ranking $c$, thus reliving $1$ from having to compromise. We can rule this out with the following condition. For any two allocations $x$ and $y$, let $d(x,y)=\{i\,:\,x_i\neq y_i\}$.

\begin{definition}
    The local compromiser assignment $\alpha$ satisfies \textbf{forward consistency} if for any allocations $x$ and $y$ such that $d(x,y)\subset \alpha(x)$, then $\alpha(y)\supset \alpha(x)-d(x,y)$.
\end{definition}

In words, suppose a subset of agents is required to compromise at allocation $x$. If we move to another allocation  $y$ where only a strict subset of these agents have compromised, those that have not yet compromised are still required to. This condition is interesting in its own right. Both top trading cycles and deferred acceptance satisfy forward consistency. Furthermore, both mechanisms can be implemented in many different ways. In deferred acceptance, for instance, one can achieve the same outcome as DA by having a single rejection each round. Similarly, the top trading cycles algorithm can be implemented by an algorithm which executes a single cycle in each round. These are both consequences of the following fact. 

\begin{proposition}\label{prop: forward consistency}
    Fix a constraint $C$. Let $\alpha$ be an implementable local compromiser assignment that satisfies forward consistency. If $\alpha'\subset \alpha$ is a local compromiser assignment, then $LP_{\alpha'}=LP_{\alpha}$.
\end{proposition}

 Note that in the proposition above, $\alpha'$ is still required to be nonempty on $\bar{C}$ since it is a local compromiser assignment for $C$. The following example demonstrates the second way in which a local compromiser assignment will fail to induce a group strategy-proof mechanism. 

\begin{example}
Consider a setting with three agents, three objects and a set of infeasible allocations together with a local compromiser assignment as shown in the figure below.
    \begin{figure}[H]
        \centering
        \includegraphics[width=0.9\linewidth]{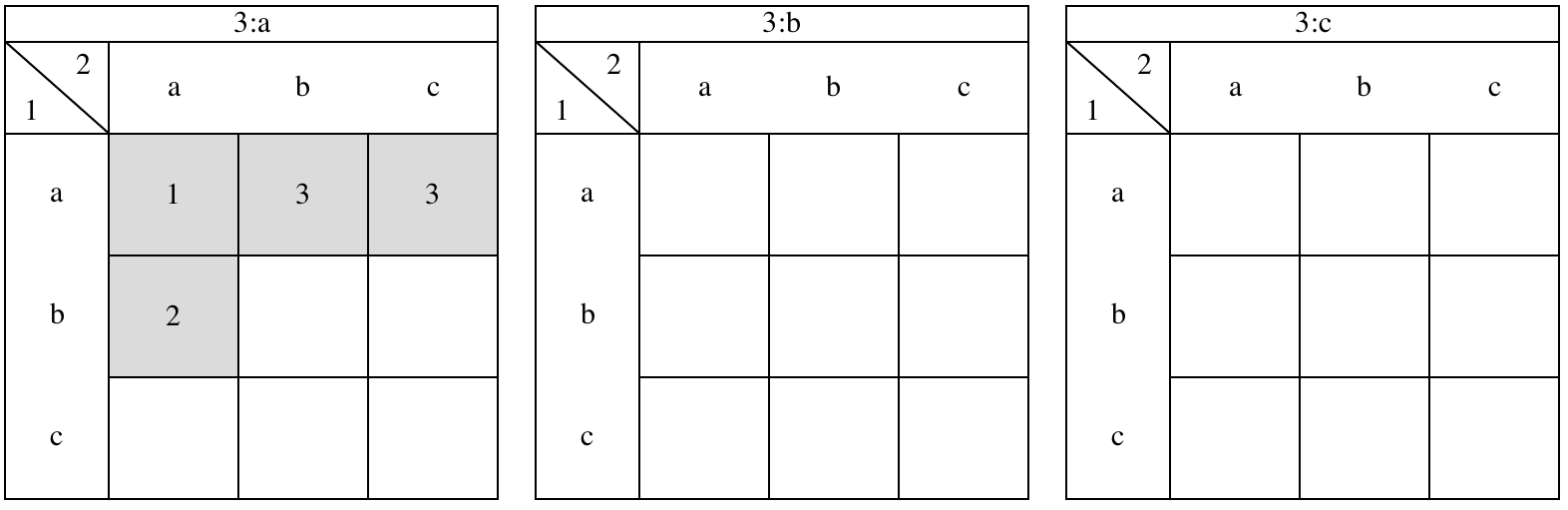}
        \label{fig: GSP violation 2}
    \end{figure}
Now consider the two profiles
   \begin{center}
        \begin{tabular}{ccccccc}
        $\succ_{1}$ & $\succ_{2}$ & $\succ_{3}$ & $\hphantom{\ldots}$ & $\succ_{1}$ & $\succ_{2}'$ & $\succ_{3}$ \\
        \cmidrule(lr){1-1}
        \cmidrule(lr){2-2}
        \cmidrule(lr){3-3}
        \cmidrule(lr){5-5}
        \cmidrule(lr){6-6}
        \cmidrule(lr){7-7}
        $a$ & $a$ & $\textcolor{blue}{a}$ &  & $\textcolor{blue}{a}$ & $\textcolor{blue}{b}$ & $a$\\
        $\textcolor{blue}{b}$ & $\textcolor{blue}{b}$ & $b$ &  & $b$ & $c$ & $\textcolor{blue}{b}$\\
        $c$ & $c$ & $c$ &  & $c$ & $a$ & $c$\\
        \end{tabular}
    \end{center} 
In the first, agent $1$ is required to compromise in the first step of the local priority algorithm. This leads to the allocation $(b,a,a)$ where $2$ is the local compromiser. After $2$ compromises we are left with the feasible allocation $(b,b,a)$. However, agents $1$ and $2$ can form a coalition and jointly misreport. Agent $2$ can instead announce $\succ_2'$. This does not change their allocation, but now in the first step of the algorithm, $3$ is required to compromise, leading to the allocation $(a,b,b)$ which is feasible. This violates group strategy-proofness.
\end{example}

The issue here is that if agent $2$ recognizes they are going to compromise and preempts this by submitting a profile where they already rank that object at the bottom, they can manipulate who compromises with them. At $\succ_2$, agent $1$ compromises with them, while at $\succ_2'$ it is agent $3$. We need a condition which implies that an agents cannot misreport to change the set of compromisers without affecting their own allocation.

Before stating the condition, it will be useful to give a few simple definitions. First, we define a sequence of allocations to be acyclic of no agent cycles through objects.

\begin{definition}
    A sequence of objects $a^1,\dots, a^p$ is \textbf{acyclic} if there is no $r<s<t$ such that $a^r=a^t$ and $a^s\neq a^r$. A sequence of allocations $x^1,\dots x^p$ is \textbf{acyclic} if for all $i$, the sequence of objects $x^1_i,\dots x^p_i$ is acyclic. 
\end{definition}

The next condition specifies conditions under which we could start from a given allocation $x$ and eventually land on $y$ after having $i$ compromise. 

\begin{definition}
    Given a local compromiser assignment $\alpha$, two infeasible allocations $x$ and $y$ are $i$\textbf{-connected} under $\alpha$ if there is an acyclic sequence of allocations $z^1,\dots z^p$ such that 
    \begin{enumerate}
        \item $z^1=x$ and $z^p=y$
        \item $d(z^1,z^2)=\{i\}$
        \item For all $l<p$, $d(z^{l+1},z^{l})\subset \alpha(z^l)$
    \end{enumerate}
\end{definition}

\begin{definition}
    A local compromiser assignment $\alpha$ is \textbf{backward consistent} if for any agent $i$ and any two infeasible allocations $x$ and $y$ that are $i$-connected, then for any $x'$ such that  $d(x,x')\subset \alpha(y)$ and $i\notin d(x,x')$ we have  $i\in \alpha(x')$.
\end{definition}

Going back to the previous example, $(a,a,a)$ and $(b,a,a)$ are $1$-connected. Since $2$ is the compromiser at $(b,a,a)$, backward consistency then requires that $1$ be among the compromisers at all $(a,d,a)$ which is not satisfied. Backwards consistency rules out situations where an agent who will need to compromise at some point can manipulate who compromises with them by preemptively compromising. Any local compromiser assignment that satisfies forward and backward consistency is called \textbf{consistent}.

\begin{theorem} \label{local priority Theorem}
Fix a constraint $C$ and an implementable local compromiser assignment $\alpha$. If $\alpha$ is consistent, then $LP_{\alpha}$ is group strategy-proof.
\end{theorem}

Our approach in proving Theorem \ref{local priority Theorem} is to show that for any consistent $\alpha$, the mechanism $LP_{\alpha}$ has marginal submechanisms that are local priority mechanisms which are also consistent. That is, if we partition the set of agents into two groups $M$ and $-M$ and fix the preferences of the agents in $-M$, say at $\succ_{-M}$, then the marginal mechanism $\succ_M \, \mapsto LP_{\alpha}(\succ_M,\succ_{-M})\vert_{M}$ is also a local priority mechanism which can be implemented by a local compromiser assignment which is consistent. This holds in particular if $M$ is any pair of agents. Now, for a pair of agents, consistency implies that the marginal mechanism is strategy-proof and efficient with respect to its image.\footnote{That is, the marginal mechanism is efficient for the constraint which is exactly its image.} As a consequence, there is no potential for the pair to jointly misreport, establishing the result (recall that it is necessary and sufficient to rule out misreports by pairs). 

We show in Appendix \ref{appendix: consistency theorem} that forward consistency is a necessary condition. We provide an example which shows that backward consistency is not necessary. Nevertheless, backwards consistency is satisfied by many prominent mechanisms including top trading cycles. We stress that Theorem \ref{local priority Theorem} can be applied to any constraint. In section \ref{sec: illustration}, we apply consistency to recover well-known mechanisms for the house allocation problem and to generate new group strategy-proof mechanisms for a simple variant of the house allocation constraint.

\section{Comparisons Across Mechanisms}\label{sec: comparisons}

In this section we compare local priority mechanisms within and across constraints holding fixed the set of agents and the set of objects. We find that our consistency conditions are necessary to generate intuitive comparative statics results. The first result states that if $\alpha \subset \alpha'$ pointwise and if $\alpha'$ satisfies forward consistency, then all agents will be weakly better off under the local priority mechanism under $\alpha$ than under the local priority mechanism under $\alpha'$. Next, we consider the performance across local priority mechanisms for a single agent. We give an example which demonstrates that agents can prefer local compromiser assignments where they compromise more to others where they compromise less. However, holding fixed the other set of compromisers and requiring consistency, agents always prefer to compromise less.

First, without forward consistency, it is possible for all agents to prefer the outcome of a local priority mechanism under $\alpha'$ to the outcome of a local priority algorithm under $\alpha$ even if $\alpha\subset \alpha'$ pointwise. To see this, suppose that there are three agents $1,2,3$ and three objects $a,b,c$. Suppose that $(b,a,a)$ and $(a,a,a)$ are the only infeasible allocations. Let $\alpha'((a,a,a))=\{1,2\}$, $\alpha'((b,a,a))=\{1,2,3\}$ and $\alpha((a,a,a))=\{1\}$, $\alpha((b,a,a))=\{1,2,3\}$. At the profile where all agents rank $a$ above $b$ above $c$, the outcome of the local priority algorithm under $\alpha'$ is $(b,b,a)$ while under $\alpha$ it is $(c,b,b)$.

If, however, $\alpha'$ satisfies forward consistency, we recover the anticipated comparative statics result.

\begin{proposition}\label{prop: CS1}
    Suppose that $\alpha$ and $\alpha'$ are implementable, $\alpha'$ satisfies forward consistency and $\alpha\subset \alpha'$. Then for any profile $\succ$ we have $LP_{\alpha}(\succ)(i)\succsim_i LP_{\alpha}(\succ')(i)$ for all $i$.
\end{proposition}

Note that $\alpha$ and $\alpha'$ are arbitrary maps from $\mathcal{A}\rightarrow 2^N$, so are not necessarily local compromiser assignments for the same constraint. 

\begin{proof}
    By Proposition \ref{prop: forward consistency}, there is some $\alpha''\subset \alpha'$ such that $\alpha''=\alpha$ whenever $\alpha$ is nonempty such that $LP_{\alpha''}=LP_{\alpha'}$.
    Then the only difference between $\alpha$ and $\alpha''$ is that there are potentially some $x$ where $\alpha(x)=\emptyset$ and $\alpha''(x)\neq \emptyset$.
    Fix a profile $\succ$. Following the local priority algorithm under $\alpha$ and $\alpha'$, either the allocations considered include no such $x$, in which case the outcome of the two algorithms is identical, or such an $x$ is considered and $LP_{\alpha}(\succ)=x$ while the local priority algorithm under $\alpha'$ considers at least one other allocation. Since welfare is decreasing in each step of the the local priority algorithm, $LP_{\alpha}(\succ)_i=x_i \succsim_i LP_{\alpha''}(\succ)_i=LP_{\alpha'}(\succ)_i$ for all $i$.
\end{proof}

As an example of Proposition \ref{prop: CS1}, increasing school capacities while holding fixed priorities makes all agents weakly better off in DA.

The next example shows that, even if $\alpha$ and $\alpha'$ are such that $i\in \alpha'(x)$ implies  $i\in \alpha(x)$, so that whenever $i$ is called to compromise using $\alpha'$, they are also called to do so in $\alpha$, they may still prefer the outcome under $\alpha$. That is, an agent may actually be better off by being made to compromise.

\begin{example}
    In Figure \ref{fig:comparisons} there are two mechanisms, each for two agents. In the first, agent $1$ is the local compromiser everywhere. In the second, $2$ is swapped for $1$ at the allocation $(a,a)$. If preferences are such that both agents prefer $a$ to $b$ to $c$, $1$ prefers their outcome in the first mechanism, $b$, to their outcome of the second mechanism, $c$. 
    \begin{figure}
        \centering
        \includegraphics[width=0.6\linewidth]{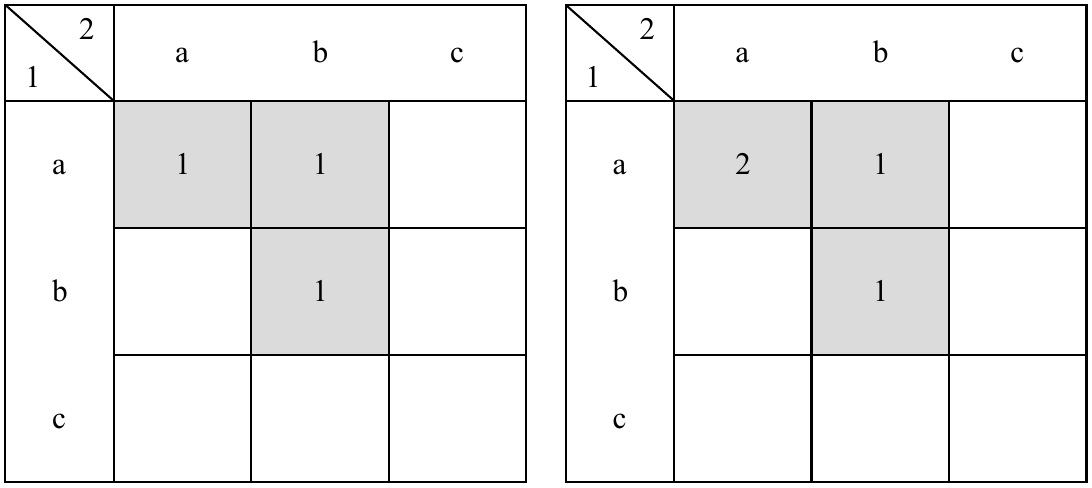}
        \label{fig:comparisons}
        \caption{A simple example where welfare is non-monotonic in the number of times an agent has to compromise}
    \end{figure}
\end{example}

The second mechanism in Figure \ref{fig:comparisons} doesn't satisfy backward consistency. If we require both $\alpha$ and $\alpha'$ to be consistent, however, we get an unambiguous improvement for $i$ under $\alpha'$

\begin{proposition}\label{prop:comparative statics}
    Fix a constraint $C$ and an agent $i$. Suppose that $\alpha$ and $\alpha'$ are two consistent, implementable local compromiser assignments for $C$ such that for every infeasible $x$:
    \begin{enumerate}
        \item for any $j\neq i$ if $j\in \alpha(x)$ then $j\in \alpha'(x)$
        \item if $i\in \alpha'(x)$ then $i\in \alpha(x)$
    \end{enumerate}
    Then $i$ weakly prefers the outcome in the local priority mechanism under $\alpha'$ to the outcome in the local priority mechanism under $\alpha$ for any preference profile. 
\end{proposition}

\section{A Three Agent, Three House Example}\label{sec: illustration}

We now turn to an application of Theorem \ref{local priority Theorem} to a simple variant of the house allocation problem. \citeasnoun{PyUn17} characterized the group strategy-proof and Pareto efficient mechanisms for the house allocation problem. With three agents and three objects, the house allocation constraint can be visualized as in Figure \ref{Three Agent Constraint}. In this section, we will make a slight perturbation to this constraint. With this small perturbation, existing analyses of the house allocation problem are inapplicable. However, by constructing local priority mechanisms for the constraint, we can find nontrivial group strategy-proof and Pareto efficient mechanisms for this problem. This is meant to concretely illustrate how local priority mechanisms can be constructed for reasonable problems for which no existing broad class of known mechanisms would apply. Moreover, the resulting mechanisms are of some interest on their own, since they describe how tensions between property rights and efficiency are adjudicated by the mechanism in the present of a slightly relaxed constraint.

\begin{figure}[h]
    \centering
    \includegraphics[scale=0.42]{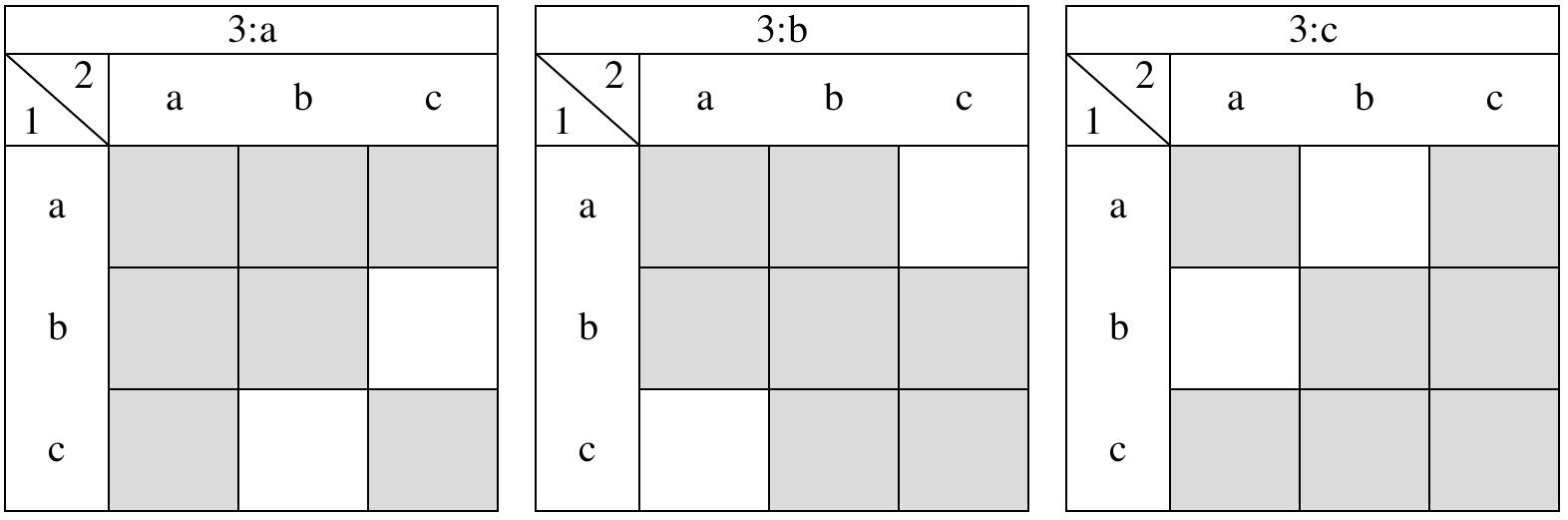}
    \caption{The constraint for the house allocation problem}
    \label{Three Agent Constraint}
\end{figure}

In Figure \ref{Three Agent House Allocation Mechanisms}, we list (up to a relabeling of the agents and objects) the set of implementable local compromiser assignments which satisfy both consistency conditions for the house allocation problem. We drop the labels above to make the figure more compact. When, for example we list ``1 or 2,'' we mean that the cell can be filled with a ``1'' or a ``2,'' but not both. It turns out that this is exactly the set of group strategy-proof and Pareto efficient mechanisms characterized by \citeasnoun{PyUn17}. There are ``hierarchical exchange" mechanisms as in \citeasnoun{Papai00} and ``broker" mechanisms as in \citeasnoun{PyUn17}. The second panel on the left is the ``braiding" mechanism of \citeasnoun{Bade16}.

\begin{figure}[h]
    \centering
    \includegraphics[scale=0.66]{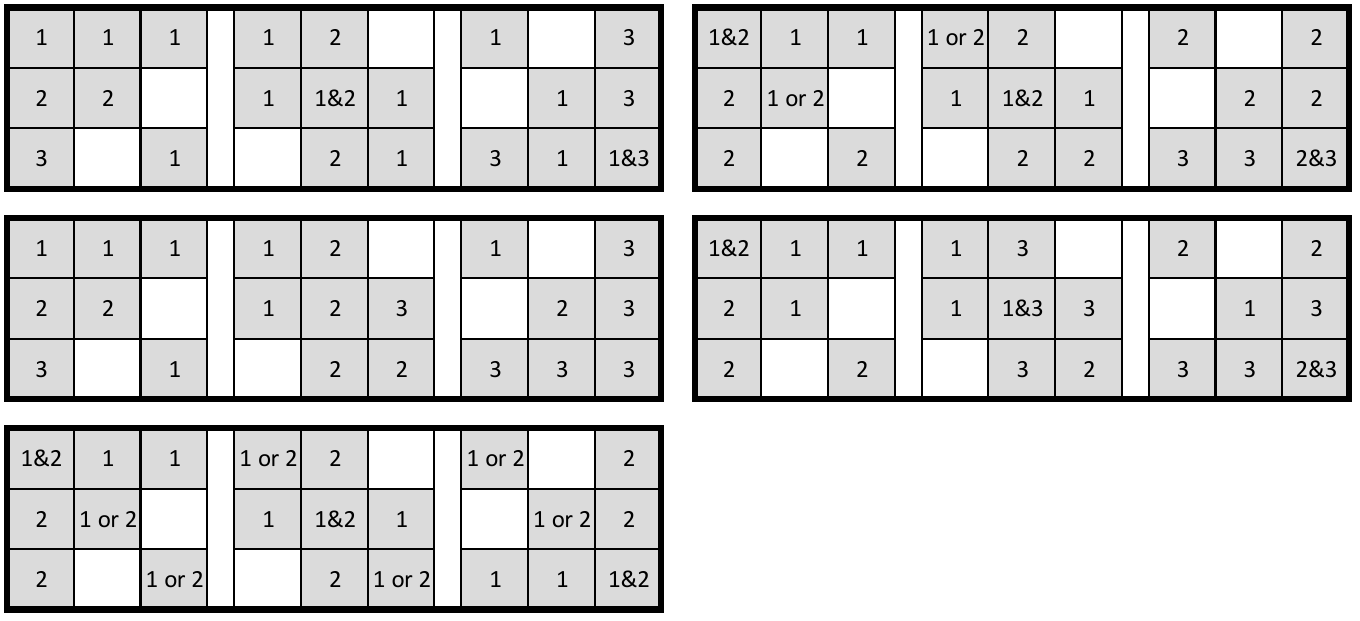}
    \caption{All group strategy-proof and Pareto efficient 3-Agent House Allocation Mechanisms (up to symmetry)}
    \label{Three Agent House Allocation Mechanisms}
\end{figure}

Suppose, however, that the constraint relaxes, as shown in Figure \ref{Three Agent Constraint Variation}. Now, the allocation $(a,a,a)$ is feasible. Otherwise the constraint is exactly the same. Without Theorem \ref{local priority Theorem}, one would need to find a way to modify the challenging proof of \citeasnoun{PyUn17} to fit this constraint. In light of this theorem, however, we can simply find the set of local compromiser assignments which satisfy forward and backward consistency. These are listed in Figure \ref{Three Agent Constraint Variation Mechs}. 
\begin{figure}[H]
    \centering
    \includegraphics[scale=0.42]{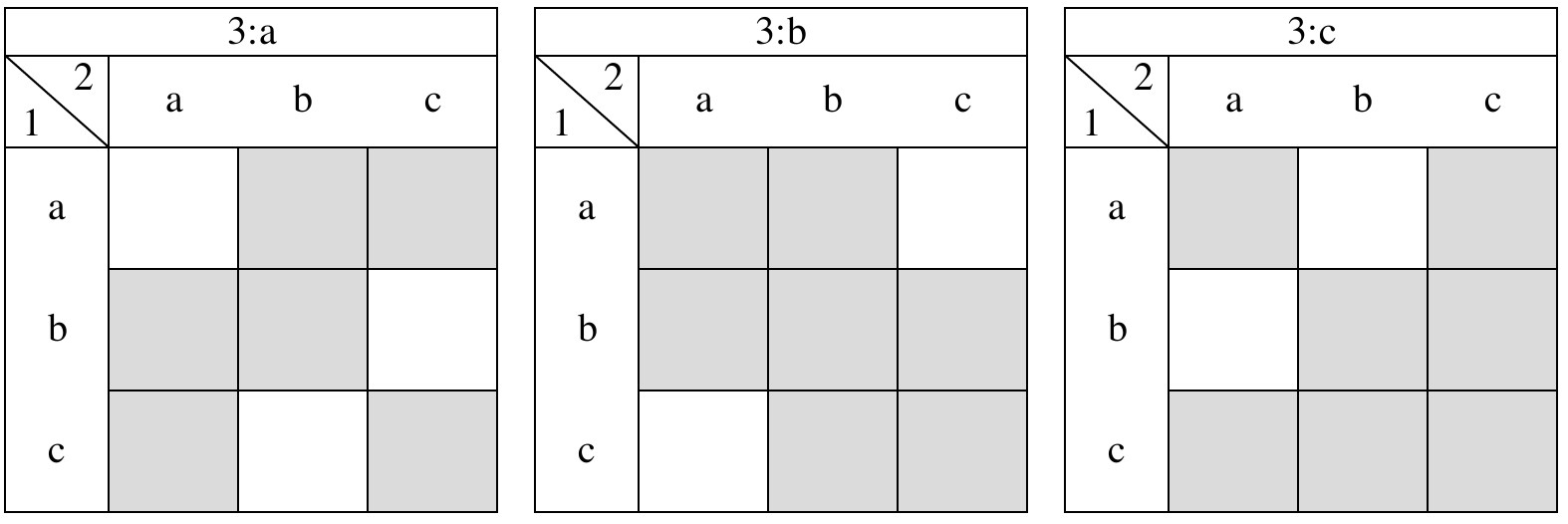}
    \caption{A variation on the house allocation constraint}
    \label{Three Agent Constraint Variation}
\end{figure} 

\begin{figure}[h]
    \centering
    \includegraphics[scale=0.66]{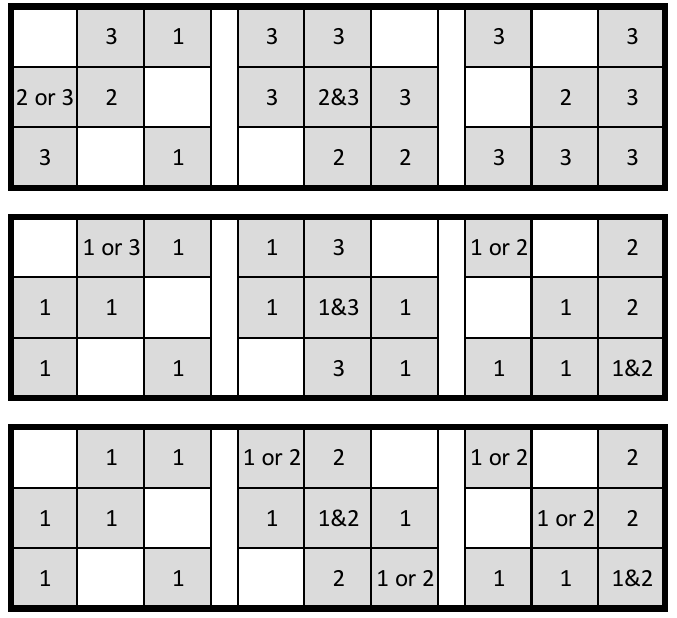}
    \caption{The consistent local priority mechanisms for the constraint in Figure \ref{Three Agent Constraint Variation}}
    \label{Three Agent Constraint Variation Mechs}
\end{figure}

These mechanisms demonstrate many of the qualities observed in house allocation problems. In the first mechanism, we have a broker as in \citeasnoun{PyUn17}. The second mechanism is a mix between serial dictatorship and top trading cycles. The mechanism behaves as though agent $2$ owns object $b$ and agent $3$ owns object $c$. If both agents $2$ and $3$ top-rank $a$, then the social allocation is $(a,a,a)$ regardless of $1$'s preferences. However, $1$ can also have some power. If we opt for $3$ in the square labeled ``$1$ or $3$" and $2$ in the square labeled ``$1$ or $2$", whenever either $2$ or $3$ top-ranks the object she owns and there is a conflict between the other two agents over $a$, in this case $1$ forces the other agent to compromise. This illustrates that simply following the consistency conditions to construct local priority algorithms can yield mechanisms with interesting properties.

\section{Discussion and Further Directions}\label{sec: discussion}

All group strategy-proof house allocation mechanisms are local priority mechanisms \cite{PyUn17}. \citeasnoun{bade2023royal} characterize the set of group strategy-proof two-sided matching mechanisms under an additional symmetry requirement. All these mechanisms are local priority mechanisms. These examples suggest that in many settings, local priority mechanisms are the most likely to yield strong incentive properties. Indeed, Corollary \ref{cor: gsp} suggests that one only needs to establish that the fixed compromiser condition holds to conclude that all unanimous and group strategy-proof mechanisms are local priority mechanisms. Both of our examples of prominent mechanisms which are not local priority mechanisms have well-known violations of strategy-proofness. Further work is needed to clarify when this can be expected to hold. 

While we have primarily studied group strategy-proofness, many important individually strategy-proof mechanisms are local priority mechanisms. In \citeasnoun{RoAh23} we gave an explicit characterization of the entire class of strategy-proof and efficient mechanisms for any constraint with just two agents. All such mechanisms are local priority mechanisms. Deferred acceptance is also a local priority mechanism. Future research should identify conditions under which a local priority mechanism is individually strategy-proof and efficient. 

In recent years, there has been interest in dynamic implementations of various mechanisms \cite{bo2020iterative}. Local priority mechanisms have an obvious dynamic implementation where each agent is asked to report a single desired object in each round. This approach may be easier for participants who need to engage in costly learning to uncover their preferences since students need only find the object which they top-rank in each round, rather than generate a full preference ranking. 

We have only considered deterministic mechanisms here, however local priority mechanisms can be used to generate random mechanisms with desirable properties. For example, randomizing over the set of group strategy-proof mechanisms for any constraint will yield a ordinally strategy-proof and ex-post efficient mechanism. In general, there is no guarantee that these mechanisms will be ex-ante efficient. However, \citeasnoun{echenique2022efficiency} discuss scenarios where ex-post and ex-ante efficiency coincide.

\section{Appendix}

\small

\subsection{Preliminary Observations}\label{appendix: preliminary}

It will be useful to relate group strategy-proofness with two other notions: nonbossiness and Maskin monotonicity.

\begin{definition}
A mechanism $f:\pp \rightarrow C $ is
\begin{enumerate}
    \item \textbf{nonbossy} if, for all $\succ\in \pp$,  $$f_{i}(\succ'_{i},\succ_{-i})= f_{i}(\succ) \implies f(\succ'_{i},\succ_{-i})= f(\succ). $$
    \item \textbf{Maskin monotonic} if, for all  $\succ, \succ'\in \pp$, 
    $$LC_{\succ_{i}'}(f_{i}(\succ))\supset LC_{\succ_{i}}(f_{i}(\succ))\text{  for all }i \implies f(\succ')=f(\succ).$$
\end{enumerate}
\end{definition}

\begin{proposition}[Root and Ahn, 2023] \label{GSP equivalences}
If $f:\pp\rightarrow \alloc$\ the following are equivalent:
\begin{enumerate}[leftmargin=*]
    \item $f$ is group strategy-proof
    \item $f$ is strategy-proof and nonbossy
    \item $f$ is Maskin monotonic.
\end{enumerate}
\end{proposition}

Note that a corollary of this fact is that it is necessary and sufficient to rule out misreports by coalitions of size at most two.

\begin{corollary}\label{cor: two agent}
    A mechanism $f$ is group strategy-proof if and only if it is individually strategy-proof and there is no profile $\succ$, pair of agents $i,j$ and alternative $\succ_{\{i,j\}}'$ such that  
    \begin{enumerate}
    \item $f_{k}(\succ'_{\{i,j\}},\succ_{-})\succsim_{j}f_{j}(\succ) \text{  for all  }k\in \{i,j\} $ and
    \item $f_{k}(\succ'_{\{i,j\}},\succ_{-})\succ_{k}f_{k}(\succ)$ for at least one $k\in \{i,j\}$.
\end{enumerate}
\end{corollary}

\begin{proof}
    If $f$ is group strategy-proof, it is individually strategy-proof and no such pair can exist. If it is individually strategy-proof and no such pair can exist, the mechanism is non-bossy so by Proposition \ref{GSP equivalences}, the mechanism is group strategy-proof.
\end{proof}

\subsection{Proof of Proposition \ref{nec suff conditions}}
 Suppose that a feasible mechanism $f$ satisfies the fixed compromiser and compromiser invariance properties. Let $C$ be the image of $f$. Define $\alpha:\alloc \rightarrow 2^N$ by setting $\alpha(x)=\emptyset$ for any $x\in C$ and for any $y\in \bar{C}$, let $\alpha(y)=\{i : f(\succ)(i)\neq y_i \text{ for all } \succ \,\in P^{x}\}$. By the fixed compromiser condition, $\alpha(y)$ is nonempty. We will show that $f(\succ)=LP_{\alpha}(\succ)$ for all $\succ$.  Fix a profile $\succ$. Construct a sequence of preference profiles as follows. First, set $\succ^0=\succ$. If $\tau_1(\succ^0)\in C$, stop. Otherwise, for each subsequent step $t\geq 1$, let $\succ^{t}$ be the profile derived from $\succ^{t-1}$ by pushing $\tau_1(\succ^{t-1})(i)$ to the bottom of $i$'s preference for all $i\in \alpha(\tau_1(\succ^{t-1}))$. If the resulting profile $\succ^t$ is such that $\tau_1(\succ^t)\in C$, stop. Otherwise, proceed to the next step. This results in a sequence of profiles $\succ^0,\dots, \succ^m$. By compromiser invariance $f(\succ^0)=\cdots = f(\succ^m)$. By the fixed compromiser condition, $f(\succ^m)=\tau_1(\succ^m)$. Furthermore, the sequence $\tau_1(\succ^0),\dots, \tau_1(\succ^m)$ is exactly the sequence of allocations considered by the local priority algorithm for $\alpha$.

 Going the other direction, it is immediate that local priority mechanisms satisfy the fixed compromiser condition. To establish compromiser invariance,  fix an implementable $\alpha$ and the associated local priority mechanism $LP_{\alpha}$. Given $\succ$ and $\succ'$ as in the compromiser invariance condition, if $x^0,x^1,\dots, x^m$ is the sequence of allocations considered in the local priority algorithm under $\alpha$ at $\succ$, $x^1,\dots, x^m$ is the sequence considered in the local priority algorithm under $\alpha$ at $\succ'$.\qed

\subsection{Proof of Proposition \ref{proposition: TTC} }\label{appendix: TTC}
Fix an initial allocation $\mu_0$.
By Proposition \ref{nec suff conditions}, it is sufficient to establish the fixed compromiser and compromiser invariance conditions. 

Fix a vector of top-choices $\mu$. If $\mu$ is feasible (i.e. all agents top-rank a different object), then in the first step of the TTC algorithm for any profile in $P^{\mu}$, all agents are involved in a cycle so that $TTC(\succ)=\mu$ as required. Now suppose that $\mu$ is infeasible so that at least two agents top-rank the same object. Pick any agent $i$ such that no agent top-ranks their initial endowment. For each such agent form a chain $i_0, i_1, \dots, i_r$ so that (a) $i_0=i$, (b) for every $l<r$, $i_l$ top-ranks the endowment of $i_{l+1}$ and (c) $i_r$ top ranks $i_p$ for $p<r$. This creates a disjoint collection of chains which includes all agents. The set of agents $i_{p-1}$ who top-rank the same object as the agent at the end of the chain are fixed compromisers since in the first step of the TTC algorithm, this object will leave with $i_p$.

Now, we want to establish compromiser invariance, or that $f(\succ) = f(\succ_M', \succ_{-M})$ for the compromisers identified in the previous paragraph where $\succ_i'$ bottom-ranks $\tau_1(\succ_i)$ and ranks everything other object the same for all $i\in M$. Let $\nu=f(\succ_{M}',\succ_{-M})$ and $\mu=f(\succ)$. Since TTC implements the unique core allocation, we need to show that $\mu$ is still a core allocation at the profile $(\succ_{M}',\succ_{-M})$. Suppose, by way of contradiction, that there is a cycle $i_1,\dots, i_r$ of agents such that each for each $l<r$, agent $i_l$ prefers prefers the object $\mu(i_{l+1})$ to $\mu(i_l)$ and $i_r$ prefers $\mu(i_1)$ to $\mu(i_r)$ at the profile $(\succ_{M}',\succ_{-M})$. This cycle must include an agent $i$ since $\mu$ is core under $\succ$ and the agents $M$ are the only ones whose preferences have changed. Furthermore, all trades executed under the first step of TTC under $\succ$ are also executed in the first step of TTC under $(\succ_{M}',\succ_{-M})$. As a consequence, for any agent $i\in M$ we have that $\nu(i)\neq \tau_1(\succ_i)$. Then the cycle $i_1,\dots, i_r$ is also a cycle of $\nu$ which is core, establishing the desired contradiction.
\qed

\subsection{Example of a Non-unique Local Compromiser Assignment}\label{appendix: non-uniqueness example}

\begin{example}
    The three panels in Figure \ref{local priority nonunique} correspond to different two-agent local compromiser assignments. In panels (II) and (III), 1 or 2 is listed as the sole compromiser at the allocation $(a,a)$, despite the fact that the other agent must compromise next no matter what the compromiser chooses. In panel (I), both agents are asked to compromise immediately, and will arrive at the same outcome.
    \begin{figure}[h]
    \centering
    \includegraphics[scale=0.45]{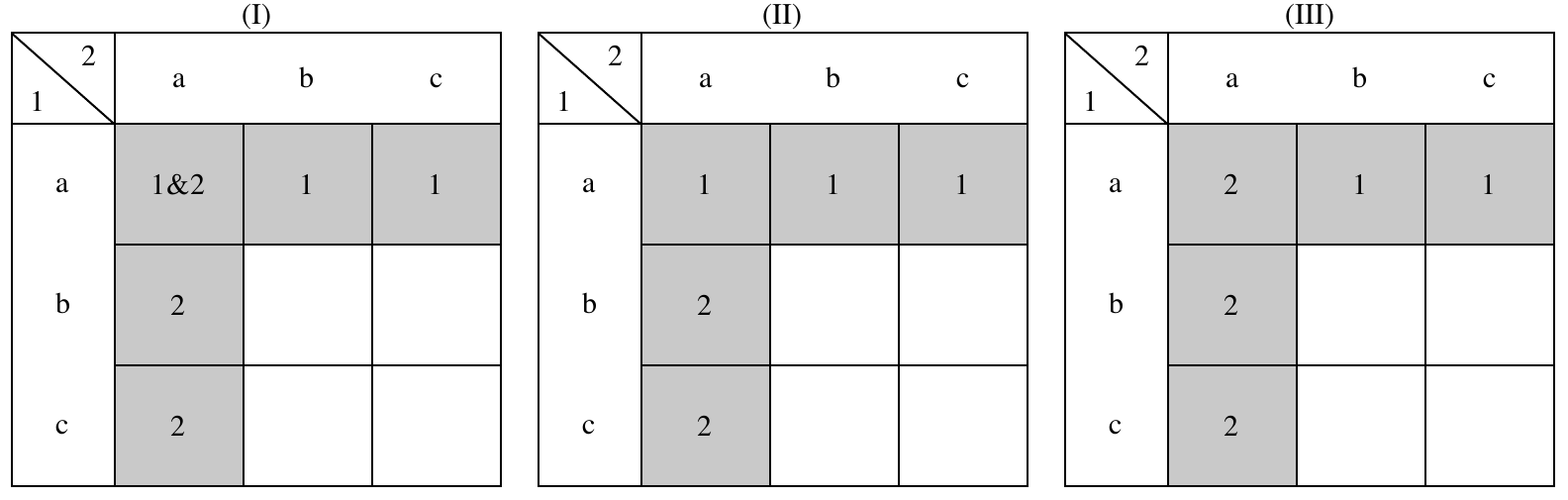}
    \caption{Three different local constraint assignments which induce the same mechanism.}
    \label{local priority nonunique}
    \end{figure}
\end{example}

\subsection{Proof of Proposition \ref{local priority PE}}

Fix some constraint $C$ and a local compromiser assignment $\alpha$ such that $LP_{\alpha}$ is group strategy-proof. By construction, any local priority mechanism's image is exactly $C$. Suppose that there is some profile $\succ$ such that $LP_{\alpha}(\succ)=\nu$ but that $\mu\in C$ is a Pareto improvement. Since $\mu$ is in the image of $LP_{\alpha}(\succ)$ there is some $\succ'$ such that $LP_{\alpha}(\succ')=\mu$. However, this gives a violation of group strategy-proofness. The mechanism shown in Figure \ref{fig: efficient, not gsp} is easily seen to be Pareto efficient. Consider the profile where all agents rank $a$ above $b$ above $c$. This mechanism results in the outcome $(b,b,a)$. However, this mechanism is bossy since if $2$ announces $b\succ_2' c \succ_2' a$, $2$ still gets $b$ but $1$ gets $a$.
\begin{figure}[h]
\centering
\includegraphics[scale=0.45]{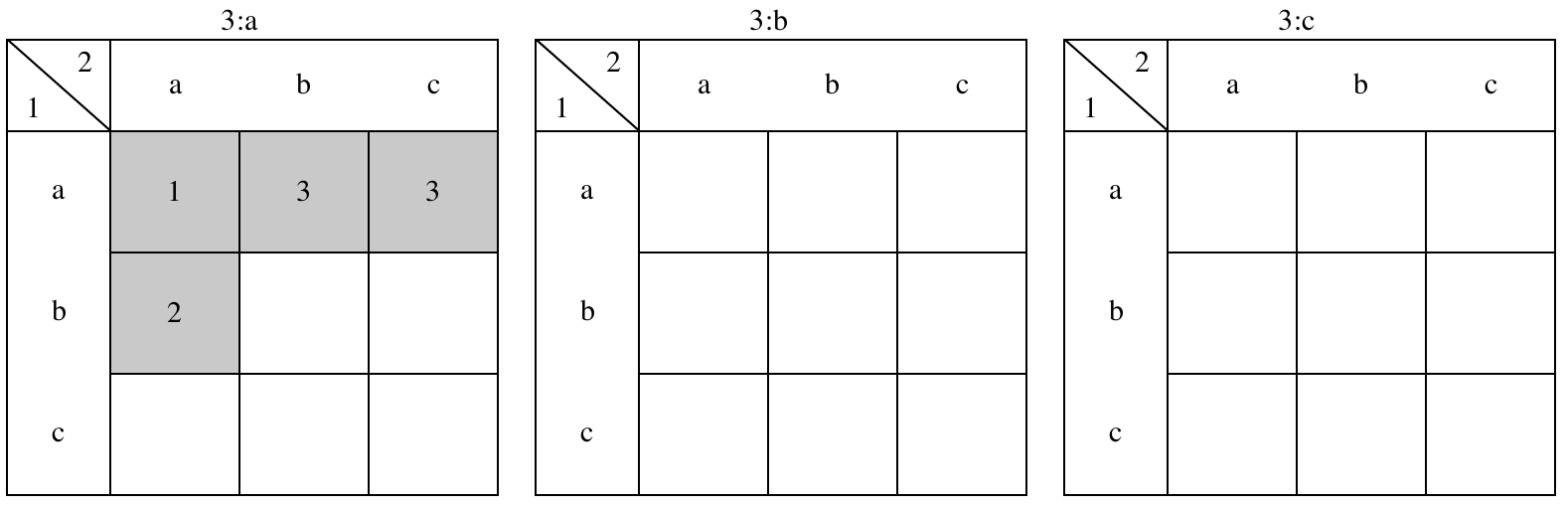}
\caption{An inefficient local priority mechanism}
\label{fig: efficient, not gsp}
\end{figure}
\qed

\subsection{Proof of Proposition \ref{local compromiser closure}}

Let $\alpha$ and $\alpha'$ induce $f$ which is group strategy-proof. Let $\succ$ be an arbitrary preference profile and set $\succ^{0}=\succ$. Iteratively define the sequence $\succ^{0},\succ^{1},\dots, \succ^{p}$ so long as $\tau_1(\succ^{n})\notin C$ by setting $\succ^{l+1}_{i}=\succ^{l}_{i}$ for all $i\notin \alpha(\tau_1(\succ^{l}))\cup \alpha'(\tau_1(\succ^{l}))$ and for all other $j$, let $\succ^{l+1}_{j}$ be identical to $\succ^{l}_{j}$ except $\tau_1(\succ_{j}^{l})$, is sent to the bottom of their list. At each step, we have $f_{j}(\succ^{l})\neq \tau_1(\succ^{l}_{j})$ for all $j\in \alpha(\tau_1(\succ^{l}))\cup \alpha'(\tau_1(\succ^{l}))$ so Maskin monotonicity implies that $f(\succ^{l})=f(\succ^{l+1})$ for each $l<p$. However the sequence $\tau_1(\succ^{l})$ is precisely the sequence of allocations considered under the local priority algorithm under $\alpha\cup \alpha'$. The algorithm ends at the first feasible assignment, and since $f(\succ^{n})$ is unchanged throughout the process, we get that $f(\succ)=f(\succ^{0})=f(\succ^{1})=\dots=f(\succ^{p})=LP_{\alpha\cup \alpha'}(\succ)$ which gives the result. 
\qed

\subsection{Proof of Proposition \ref{prop: forward consistency}}

Fix a constraint $C$ and let $\alpha$ be implementable and forward consistent. Let $\alpha'\subset \alpha$ be a local compromiser assignment such that $LP_{\alpha'}=LP_{\alpha}$. Fix an infeasible allocation $x$ such that $\vert\alpha'(x)\vert\geq 2$ and some $i\in \alpha'(x)$. Let $\alpha''(x)=\alpha'(x)-\{i\}$ and otherwise set $\alpha''=\alpha'$. We will show that $LP_{\alpha''}=LP_{\alpha}$. This will prove the proposition since we can then remove agents one at a time to go from $\alpha$ to any pointwise subset.

For any profile $\succ$ under which the local priority algorithm under $\alpha'$ never considers $x$, the local priority algorithm under $\alpha''$ also never considers $x$ so we have $LP_{\alpha'}(\succ)=LP_{\alpha''}(\succ)$ and therefore
$LP_{\alpha''}(\succ)=LP_{\alpha}(\succ)$. If, instead it does consider $x$, let $\succ'$ be the profile such that, $$x_{j}\succ'_{j} LC_{\succ_{j}}(x_{j})\succ_{j}' UC_{\succ_{j}}(x_{j})$$ where the ranking within groups is identical to $\succ_{j}$ for all $j$. Finally, let $\succ''$ be the profile derived from $\succ'$ where all agents in $\alpha'(x)-\{i\}$ put their top choice to the bottom of their lists.

The following chain of equalities will prove the result:
$$LP_{\alpha}(\succ)\stackrel{\rm A}{=}LP_{\alpha'}(\succ)\stackrel{\rm B}{=}LP_{\alpha'}(\succ')\stackrel{\rm C}{=}LP_{\alpha}(\succ')\stackrel{\rm D}{=}LP_{\alpha}(\succ'')\stackrel{\rm E}{=}LP_{\alpha'}(\succ'')\stackrel{\rm F}{=}LP_{\alpha''}(\succ'')\stackrel{\rm G}{=}LP_{\alpha''}(\succ')\stackrel{\rm H}{=}LP_{\alpha''}(\succ)$$

\noindent The equalities (A), (C) and (E) are simply a consequence of the fact that $LP_{\alpha}=LP_{\alpha'}$ by assumption.

\noindent The equalities (B) and (H) are due to the fact that the sequences of allocations considered by the local priority algorithm at $\succ'$ under $\alpha'$ and $\alpha''$ are simply truncations of those considered at $\succ$.

\noindent The equality (G) is a direct consequence of compromiser invariance.

\noindent The equality (F) is due to the fact that the local priority algorithm under $\alpha''$ at $\succ''$ does not consider $x$ so since $\alpha'=\alpha''$ outside of $x$, we get the result. 

The only remaining equality is (D), i.e. that $LP_{\alpha}(\succ')=LP_{\alpha}(\succ'')$  which we now prove. Let $y^1,\dots y^q$ be the sequence of allocations considered by the local priority algorithm under $\alpha$ at $\succ'$ and let  $z^1,\dots z^p$ be the sequence of allocations considered by the local priority algorithm under $\alpha$ at $\succ''$. For any allocation $w$, let $$\Gamma(w)=\left(\vert UCS_{\succ_i'} w_i \vert\right)_{i\in N}\in \mathbb{R}^N$$ be the vector which counts the size of the upper contour sets at $\succ'$ for the allocation $w$. By our assumptions $0=\Gamma(y^1)< \Gamma(z^1) < \Gamma(y^2)$.\footnote{For two vectors $x$ and $y$, we write $x\geq y$ if $x_i\geq y_i$ for all $i$ and $x>y$ if $x_i\geq y_i$ for all $i$ and for at least one $j$ $x_j>y_j$.}

First, we claim that if there is an integer $r$ such that $y^r=z^r$, then $p=q$ and $y^q=z^p$. Similarly, if there is an $r$ such that  $y^{r+1}=z^r$ then $q=p+1$ and $y^q=z^p$. In either case, $r\geq 2$. This means that for every $s\geq r$, $y^s_i\neq x_i$ for all $i\in \alpha(x)\supset \alpha''(x)$. However, only the rankings of the $x_i$ change between $\succ'$ and $\succ''$ for the agents in $\alpha''(x)$, and all these agents rank these objects last in $\succ''$. Thus we have, in the first case that $y^{r+1}=z^{r+1},\dots y^q=z^p$ and this last allocation is feasible. In the second case, we have $y^{r+1}=z^r, y^{r+2}=z^{r+1}, \dots, y^{p+1}=z^q$ and this last allocation is feasible. 

Now $\{i: y^1_i\neq z^1_i\}=\alpha''(x)\subsetneq \alpha(x) = \{i: y^2_i\neq y^1_i\}$ so by forward consistency, $\alpha(z^1)\supset  \{i: y^2_i\neq y^1_i\} - \{i: y^1_i\neq z^1_i\}$ and $\Gamma(y^2)\leq \Gamma(z^2)$. If $\Gamma(y^2)= \Gamma(z^2)$ we are done by the previous claim. Otherwise, $\Gamma(y^2)< \Gamma(z^2)$ so that $\alpha(z^1)=\{i: z^2_i\neq z^1_i\}\supsetneq \{i: y^2_i\neq z^1_i\}$ and by forward consistency, $\alpha(y^2)\supset \alpha(z^1)=\{i: z^2_i\neq z^1_i\}- \{i: y^2_i\neq z^1_i\}$ and $\Gamma(z^2)\leq \Gamma(y^3)$. Again, if $\Gamma(z^2)= \Gamma(y^3)$ we are done by the previous claim. Continuing in this way we get a sequence $$\Gamma(y^1)< \Gamma(z^1) < \Gamma(y^2) < \Gamma(z^2) < \Gamma(y^3) < \Gamma(z^3) \cdots .$$
However, for any $l$ such that either $\Gamma(y^l)< \Gamma(z^l)$ or $\Gamma(z^l)<\Gamma(y^{l+1})$ we have that $l\neq p$ and $\Gamma(y^l)< \Gamma(z^l)\leq \Gamma(y^{l+1})$ or $l\neq q$ and $\Gamma(z^l)<\Gamma(y^{l+1})\leq\Gamma(z^{l+1})$ respectively. Since these sequences terminate there is some $r$ such that either $\Gamma(y^r)=\Gamma(z^r)$ or $\Gamma(z^r)=\Gamma(y^{r+1})$ so that $LP_{\alpha}(\succ')=LP_{\alpha}(\succ'')$ by the previous claim. \qed

\subsection{Proof of Theorem \ref{local priority Theorem}}\label{appendix: consistency theorem}
We first show that forward consistency is necessary for a local priority mechanism to be group strategy-proof. Then we prove the theorem. Finally, we give an example to show that backward consistency is not necessary.

Fix a constraint $C$ and let $f$ be a group strategy-proof local compromiser assignment and let $\alpha$ be the union of the local compromiser assignments which implement $f$. By Proposition \ref{local compromiser closure}, $f=LP_{\alpha}$. It will be useful to have the following lemma. 

\begin{lemma}
    Suppose that $f$ is a group strategy-proof local priority mechanism and $\alpha$ is the union of the local compromiser assignments which implement it. Then for any $y$, if $i\notin \alpha(y)$ there is some $\succ\,\in P^y$ such that $f(\succ)_i=y_i$.
\end{lemma}

\begin{proof}
    Suppose that there is no such $\succ$. Let $\alpha'(y)=\{i\}$ and otherwise, let $\alpha'$ be identical to $\alpha$. We claim that $\alpha'$ is implementable and $LP_{\alpha}=LP_{\alpha'}$ which contradicts $\alpha$ being the union of all local compromiser assignments which implement $f$. Of course, the local priority algorithm for $\alpha$ and $\alpha'$ can only yield different outcomes for preference profiles $\succ$ where the local priority algorithm lands at $y$ at some point (the set of such profiles is the same for $\alpha$ and $\alpha'$). 
     Any such preference profile must satisfy $y_{j}\succsim_{j}LP_{\alpha}(\succ)_{j}$ for all $j$. 
     Given  such a profile $\succ$, define $\succ'$ so all agents $j\neq i$ have the preference $\succ'_{j}$ defined by $$y_{j}\succ'_{j} LC_{\succ_{j}}(y_{j})\succ_{j}' UC_{\succ_{j}}(y_{j})$$ where the ranking within sets is identical to $\succ_{j}$. By Maskin monotonicity, $f(\succ)=f(\succ')$. Furthermore, by assumption, $f(\succ)_i = f(\succ')_i\neq y_i$. Now, let $\succ_{i}''$ be such that $$ LC_{\succ_{i}}(y_{i})\succ_{i}'' UC_{\succ_{i}}(y_{i})\succ_{i}'' y_{i}$$ where again the ranking within the three sets is determined by $\succ_{i}$. 
     Maskin monotonicity implies that $LP_{\alpha}(\succ)=LP_{\alpha}(\succ')=LP_{\alpha}(\succ_i'',\succ_{-i}')$. Furthermore, $LP_{\alpha'}(\succ)=LP_{\alpha'}(\succ')$ since the sequence of allocations considered under $\alpha'$ at $\succ'$ is a truncation of those considered at $\succ$. Now, $LP_{\alpha'}(\succ')=LP_{\alpha'}(\succ_i'',\succ_{-i}')$ by compromiser invariance. Finally, $LP_{\alpha'}(\succ_i'',\succ_{-i}')=LP{\alpha}(\succ_i'',\succ_{-i}')$ since the local priority algorithms under $\alpha$ and $\alpha'$ at $(\succ_i'',\succ_{-i}')$ never consider $y$, the only allocation where they differ. Putting this together, we have $$ LP_{\alpha}(\succ)=LP_{\alpha}(\succ')=LP_{\alpha}(\succ_i'',\succ'_{-i})=LP_{\alpha'}(\succ_i'',\succ'_{-i})=LP_{\alpha'}(\succ')=LP_{\alpha'}(\succ).$$ Since $\succ$ was an arbitrary profile where the local priority algorithm under $\alpha$ lands on $y$, we have that $\alpha'$ implements $f$, contradicting $\alpha$ as the union of all such local compromiser assignments. 
\end{proof}

\paragraph{Necessity of forward consistency:}
Fix a group strategy-proof local priority mechanism $f$ for some constraint $C$ and let $\alpha$ be the union of all local compromiser assignments which implement $f$.  Fix an infeasible $x$ such that there is a nonempty subset $A\subsetneq\alpha(x)$. Let $y$ be an allocation where $y_j=x_j$ for all $j\notin A$. If $y$ were feasible, then $LP_{\alpha}$ would be inefficient since for any profile $\succ\,\in P^x$ we would have that $y$ strictly Pareto improves $LP_{\alpha}(\succ)$. However, by Proposition \ref{local priority PE}, this is not possible so $y\notin C$. Suppose that some $i\in \alpha(x)-A$ is not in $\alpha(y)$. Then by the previous lemma, there is some profile $\succ\, \in P^y$ such that $f(\succ)_i=y_i$. Let $\succ'$ be identical to $\succ$ except that all agents $j$ move $x_j$ to the top of their ranking. By Maskin monotonicity, $f(\succ)=f(\succ')$, however $f(\succ')_i\neq y_i$, a contradiction. Therefore, $i\in \alpha(y)$.

\begin{definition}
Let $f:\pp\rightarrow C$ and let $M$ be a proper subset of $N$. Let $\succ_{-M}$ be a profile of preferences of agents not in $M$. The \textbf{marginal mechanism} of $f$ holding fixed $\succ_{-M}$ is denoted $f_{\succ_{-M}} :P^{M} \rightarrow C^{M}$ and defined by $$\succ_M \enspace \mapsto \, \left[f_{j}(\succ_M, \succ_{-M}) \right]_{j\in M}.$$
\end{definition}

We now turn to sufficiency. We restate the theorem for the readers convenience.

\newtheorem*{T1}{Theorem~\ref{local priority Theorem}}
\begin{T1} 
Fix a constraint $C$ and an implementable local compromiser assignment $\alpha$. If $\alpha$ is consistent then $LP_{\alpha}$ is group strategy-proof.
\end{T1}

We prove the result in four steps. Let $\alpha$ be an implementable and consistent local compromiser assignment for the constraint $C$. First, by Proposition \ref{prop: forward consistency} for any other local compromiser assignment $\alpha'\subset \alpha$ we have $LP_{\alpha}=LP_{\alpha'}$. Next, we show that this, along with forward consistency imply that the marginal mechanisms holding a single agents' preference fixed are all local priority mechanisms for some constraint. Third, we show that these local priority mechanisms also satisfy forward and backward consistency since $\alpha$ does. Finally, we establish the result by showing that forward and backward consistency imply the group strategy-proofness of all two-agent marginal mechanisms. Then Corollary \ref{cor: two agent} gives the result.

Let $h$ be the marginal mechanism holding some agent $k$ at a fixed $\succ_{k}$. We first want to show that $h$ is a local priority mechanism. To do this, for every $x\in \bar{C}$ define 
\[ \alpha'(x)=
\begin{cases}
\alpha(x) & \text{ if }k\notin \alpha(x) \\
\{k\} & \text{ if }k\in \alpha(x) 
\end{cases}
\]
by Proposition \ref{prop: forward consistency}, $LP_{\alpha'}=LP_{\alpha}$. Now, for any suballocation $z$ of the agents other than $k$, consider the sequence $(\tau_{1}(\succ_{k}),z), (\tau_{2}(\succ_{k}),z),\dots$. If there is an $l$ such that $(\tau_{l}(\succ_{k}),z)$ is feasible and for all $m<l$, $\alpha'((\tau_{l}(\succ_{m}),z))=\{k\}$ then set $\alpha^*(z)=\emptyset$. Otherwise, set $\alpha^*(z)=\alpha'(y)$ where $y$ is the first allocation in the sequence $(\tau_{1}(\succ_{k}),z), (\tau_{2}(\succ_{k}),z),\dots$ such that $\alpha'(y)\neq \{k\}$. Let $C^{*}=\{z : \alpha^*(z)=\emptyset\}$.

\begin{lemma}
    $h=LP_{\alpha^*}$.
\end{lemma}

\begin{proof}
    Fix a profile $\succ_{-k}$ and let $z=\tau_1(\succ_{-k})$. We will observe that the mechanism $h$ satisfies three properties:
    \begin{enumerate}
        \item $h$ exhibits unanimity with respect to $C^*$. That is, if $z\in C^*$ then $h(\succ_{-k})=z$
        \item The compromisers are consistent with $\alpha^*$. Specifically, if $z$ is not in $C^*$ then $h(\succ_{-k})_i\neq z_i$ for all $i\in \alpha^*(z)$.
        \item $h$ exhibits compromiser invariance. If $z$ is not in $C^*$ and $\succ_{-k}'$ is a profile where all agents in $\alpha^*(z)$ move their top choice to the bottom of their list, then $h(\succ_{-k})=h(\succ_{-k}')$.
    \end{enumerate}
    The first two facts follow immediately from the definition of $\alpha^*$ and given that $LP_{\alpha}=LP_{\alpha'}$. To prove the third part, we will need backward consistency. Granting that these three properties hold, the result follows from noticing that for any profile $\succ_{-k}$ we can construct a sequence of profiles $\succ^1_{-k},\succ^2_{-k}, \dots, \succ^p_{-k}$ such that $\succ^1_{-k}=\succ_{-k}$ and for any step $l$, $\succ^{l+1}_{-k}$ differs from $\succ_{-k}^l$ in that all agents from $\alpha^{*}(\tau_1(\succ_{-k}^l))$ put their top-ranked alternative to the bottom of their list. By properties $2$ and $3$ we have that $h(\succ^{l+1}_{-k})=h(\succ_{-k}^l)$ for all $l$ and the sequence $\tau_1(\succ^1_{-k}),\tau_1(\succ^2_{-k}), \dots, \tau_1(\succ^p_{-k})$ is exactly the sequence of allocations considered by the local priority algorithm under $\alpha^*$ at $\succ_{-k}$.

    It remains to show that the third property above holds. To do this, let $\succ_{-k}$ and $\succ_{-k}'$ be as described. The local priority algorithm under $\alpha'$ at the profile $(\succ_k,\succ_{-k})$ starts with $k$ compromising a finite number of times before landing on some $(\tau_r(\succ_k),z)$ where $\alpha'=\alpha^*$. For each $l<r$, $(\tau_l(\succ_k),z)$ and $(\tau_r(\succ_k),z)$ are $k$-connected. By backward consistency of $\alpha$, for each $l<r$, we have that $k\in \alpha(\tau_l(\succ_k),z')$. Hence the local priority algorithm under $\alpha'$ at the profile $(\succ_k,\succ_{-k}')$ starts by considering each $(\tau_l(\succ_k),z')$ where $k$ compromises, before eventually landing on $(\tau_r(\succ_k),z')$. From that point forward, the local priority algorithms under $\alpha'$ at $(\succ_k,\succ_{-k})$ and $(\succ_k,\succ_{-k}')$ are identical. That is, only the first $r$ steps differ between profiles $(\succ_k,\succ_{-k})$ and $(\succ_k,\succ_{-k}')$. This establishes the third property and the Lemma.
\end{proof}

\begin{lemma}
    $\alpha^*$ satisfies forward consistency.
\end{lemma} 

\begin{proof}
    Fix some $z$ and $z'$ such that $d(z,z')\subsetneq \alpha^*(z)$. Let $r$ be the first natural number such that $\alpha'(\tau_r(\succ_k),z)\neq \{k\}$. By definition, $\alpha^*(z)=\alpha'(\tau_r(\succ_k),z)=\alpha(\tau_r(\succ_k),z)$. By backward consistency of $\alpha$, for each $l< r$, we have that $\alpha'(\tau_l(\succ_k),z')=\{k\}$. By forward consistency of $\alpha$,
    \begin{equation*}
        \alpha(\tau_r(\succ_k),z')\supset \alpha(\tau_r(\succ_k),z)-d(z,z').
    \end{equation*} If $k\notin \alpha(\tau_r(\succ_k),z')$, then $\alpha(\tau_r(\succ_k),z') = \alpha'(\tau_r(\succ_k),z')=\alpha^*(z')$ and we are done. Otherwise, $k\in \alpha(\tau_r(\succ_k),z')$. In this case, we can let $k$ compromise and repeatedly apply forward consistency of $\alpha$ to conclude that $\alpha^*(z')\supset \alpha^*(z)-d(z,z')$ as desired.
\end{proof}

We now need to establish that $\alpha^*$ satisfies backward consistency. To that end, suppose that $z$ and $z'$ are $i$-connected in $\alpha^*$. Let $r$ be the smallest value such that $k\notin \alpha(\tau_r(\succ_k),z)$. Similarly, let $s$ be the smallest value such that $k\notin \alpha(\tau_s(\succ_k),z')$. The following lemma does most of the work in proving that $\alpha^*$ satisfies backward consistency. 

\begin{lemma}\label{lemma: backward consistency}
$\alpha$ satisfies the following two properties:
\begin{enumerate}
    \item For any $l<r$, $(\tau_l(\succ_k),z)$ and $(\tau_s(\succ_k),z')$ are $k$-connected under $\alpha$
    \item $(\tau_r(\succ_k),z)$ and $(\tau_s(\succ_k),z')$ are $i$-connected under $\alpha$.
\end{enumerate}
\end{lemma}

Before proving the lemma, we argue that it delivers backward consistency of $\alpha^*$. Let $z^*$ be a suballocation for the agents other than $k$ such that $d(z^*,z')\subset \alpha^*(z')$ and such that $i\notin d(z^*,z')$. Now, by definition, $\alpha^*(z')=\alpha(\tau_s(\succ_k),z')$ and $k\notin \alpha^*(z')$. By (1) above, using backward consistency of $\alpha$, we have that $k\in \alpha(\tau_l(\succ_k),z^*)$ for all $l<r$. By (2) above, again using backward consistency of $\alpha$, we have that $i$ is in $\alpha(\tau_r(\succ_k),z^*)$. We need to show that $i\in \alpha^*(z^*)$. However, we have shown that $\alpha'(\tau_l(\succ_k),z^*)=\{k\}$ for all $l<r$ and $i\in \alpha(\tau_r(\succ_k),z^*)$. If $k\notin \alpha(\tau_r(\succ_k),z^*)$ we are done. Otherwise, we can repeatedly apply forward consistency to remove $k$ from the set of compromisers to find that $i\in \alpha^*(z^*)$ as desired.

We now prove the lemma above.

\noindent \textit{Proof of Lemma \ref{lemma: backward consistency}:}
It suffices to show that for any acyclic sequence of suballocations $z^1,\dots, z^p$ for the agents other than $k$, we can find a non-decreasing list of natural numbers $n_1,\dots, n_{p}$ such that $n_1=r$ and $n_{p}=s$ such that the sequence of (complete) allocations 
\begin{equation*}
    \begin{split}
        (\tau_1(\succ_k),z^1),\dots, &(\tau_{n_1}(\succ_k),z^1),\\
        &(\tau_{n_1}(\succ_k),z^2), \dots, (\tau_{n_2}(\succ_k),z^2), \\
        &\hspace{2.57cm}(\tau_{n_2}(\succ_k),z^3), \dots \\
        &\hspace{4.44cm} \dots, (\tau_{n_{p-1}}(\succ_k),z^{p}), \dots ,(\tau_{n_{p}}(\succ_k),z^{p})
    \end{split}
\end{equation*}
is acyclic such that at each step the set of agents whose allocation changes is a subset of compromisers under $\alpha$. Note that in the sequence above, for any step where $k$ does not compromise, the set of compromisers are from $\alpha^*$. We need to show that this is consistent with $\alpha$. Start with $$(\tau_1(\succ_k),z^1),\dots, (\tau_{n_1}(\succ_k),z^1),(\tau_{n_1}(\succ_k),z^2).$$ Since $n_1=r$, we have that $\alpha(\tau_{n_1}(\succ_k),z^1)=\alpha^*(z^1)$. By backward consistency of $\alpha$, we have that $k\in \alpha(\tau_l(\succ_k),z^2)$ for all $l<n_1$. If $k\notin (\tau_{n_1}(\succ_k),z^2)$ then set $n_2=n_1$. Otherwise, let $n_2$ be the first number such that $k\notin\alpha(\tau_{n_2}(\succ_k),z^2)$. By definition, $\alpha(\tau_{n_2}(\succ_k),z^2)=\alpha^*(z^2)$. For the next step, we have, by backward consistency of $\alpha$ that for any $l<n_2$, $k\in \alpha(\tau_l(\succ_k),z^3)$. Let $n_3$ be the smallest number such that $k\notin\alpha(\tau_{n_3}(\succ_k),z^3)$. By definition, $\alpha(\tau_{n_3}(\succ_k),z^3)=\alpha^*(z^3)$. Continue in this way to get the desired sequence. This will be acyclic by definition since it is acyclic for all agents other than $k$, and by construction $k$ is moving down their preference list under $\succ_k$. 
\qed

We have established the following fact.

\begin{lemma}
    $\alpha^*$ satisfies backward consistency.
\end{lemma}

Finally, we may prove the result. Every 2-agent marginal mechanism can be derived by taking margins, one-at-a-time to fix all agents preferences other than the two agents. By the results above, the two-agent marginal mechanism will be a local priority mechanism which satisfies forward and backward consistency. The only two-agent mechanisms which satisfy consistency are the set of local dictatorships (\citeasnoun{RoAh23}), which are exactly the set of group strategy-proof two-agent mechanisms. \qed

Now we show that backward consistency is not necessary through an example. 

\begin{example}
Figure \ref{fig: counterexample} gives a local priority mechanism with four agents and three objects. First, we show that this example violates backward consistency. Note that $(a,a,a,a)$ and $(b,b,a,a)$ are $1$ and $2$-connected. $3$ is the local compromiser at $(b,b,a,a)$. However, $4$ is the only local compromiser at $(a,a,b,a)$ violating backward consistency. One can verify that this mechanism is group strategy-proof.
    \begin{figure}[H]
        \centering
        \includegraphics[scale=0.5]{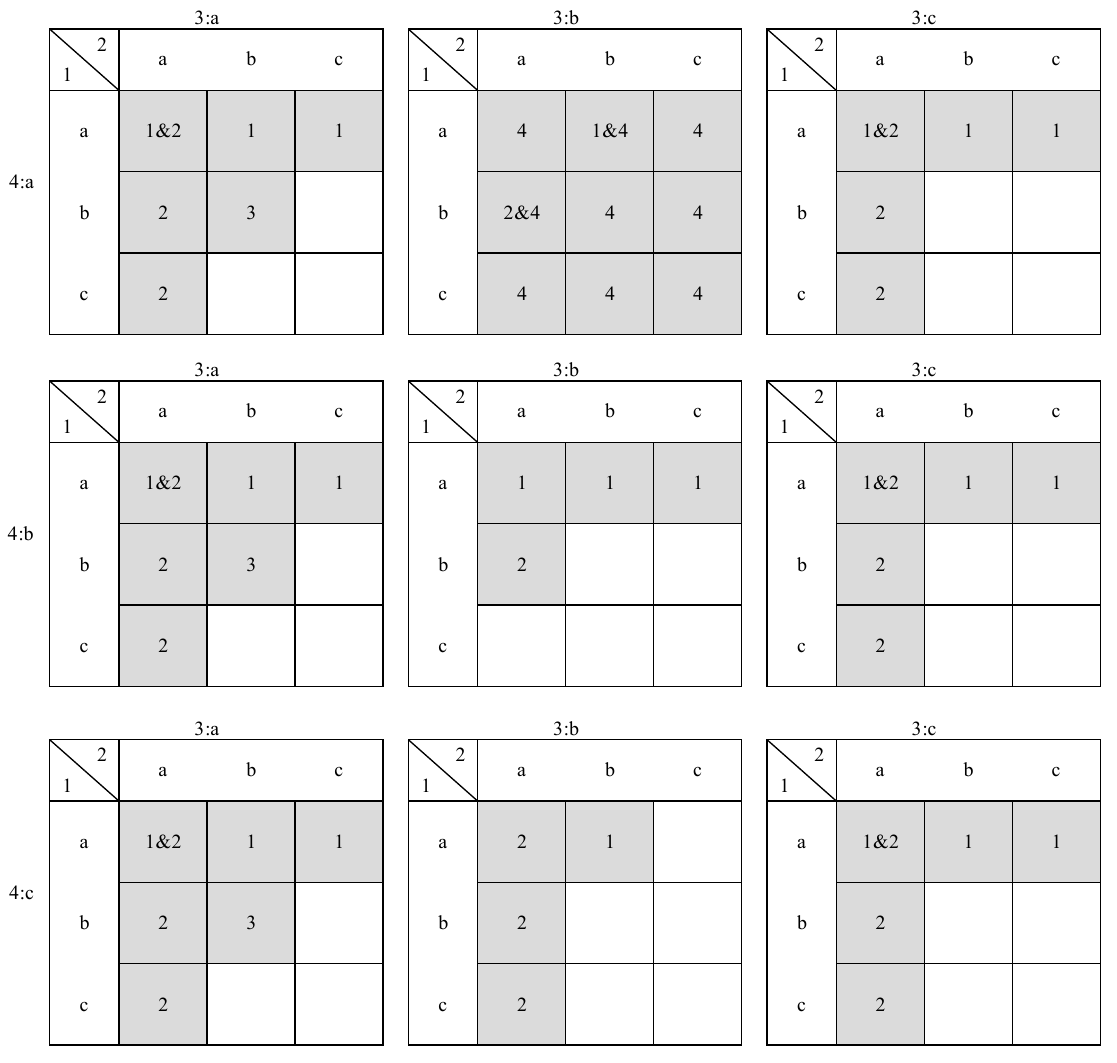}
        \caption{A local priority mechanism with 3 objects and four agents. Agent $1$'s allocation is determined by the row, agent $2$'s by the column. Agent $3$'s allocation is determined by the panel going left to right and $4$'s is determined by the panel going up and down.}
        \label{fig: counterexample}
    \end{figure}
\end{example}

\subsection{Proof of Proposition \ref{prop:comparative statics}}

Suppose that $\alpha$ and $\alpha'$ are as in the statement of the proposition. By Proposition \ref{local compromiser closure}, we can modify $\alpha$ and $\alpha'$ to $\beta$ and $\beta'$ so that
\begin{enumerate}
    \item $\beta\subset \alpha$ and $\beta'\subset \alpha'$ pointwise
    \item $\beta$ and $\beta'$ are single-valued on the infeasible set
    \item If $i$ is in $\alpha(x)$ for any $x$ then $\beta(x)=\{i\}$
    \item If $i$ is in $\alpha'(x)$ for any $x$ then $\beta'(x)=\{i\}$
    \item If $i\notin \alpha(x)$ then $\beta(x)=\beta'(x)$
\end{enumerate} and $\beta$ and $\beta'$ induce the same local priority mechanism as $\alpha$ and $\alpha'$ respectively. Fix any preference profile $\succ$ and let $z^{0}\xrightarrow{j_{1}} z^{1}\xrightarrow{j_{2}} \cdots \xrightarrow{j_{p}} z^{p}$  and $w^{0}\xrightarrow{k_{1}} w^{1}\xrightarrow{k_{2}} \cdots \xrightarrow{k_{q}} w^{q}$ be the sequence of allocations reached under the local priority mechanism for $\beta$ and $\beta'$ respectively. The two sequences are identical until the first $l$ such that $\beta(z^l)=\{i\}$ and $\beta'(w^l)\neq\{i\}$. If there is no $l'>l$ such that $\beta'(w^{l'})=\{i\}$ we have the desired result since $i$ never compromises along the sequence under $\beta'$ but does under the sequence under $\beta$. Otherwise there is such a $l'$. Modify the preference profile so that all agents put all alternatives they prefer to their allocation under $z^{l+1}$ to the bottom of their list, otherwise maintaining the rankings over the other objects. Call this profile $\succ'$. The sequence of allocations reached under the local priority mechanism for $\beta$ at $\succ'$ is simply the original sequence starting at $z^{l+1}$:  $z^{l+1}\xrightarrow{j_{l+2}} z^{l+2}\xrightarrow{j_{l+3}} \cdots \xrightarrow{j_{p}} z^{p}$. The sequence of allocations reached under the local priority mechanism for $\beta'$ at $\succ'$ may be changed, but since the local priority mechanism under $\beta'$ is group strategy-proof, the new sequence $r^{0}\xrightarrow{i_{1}} r^{1}\xrightarrow{i_{2}} \cdots \xrightarrow{i_{q}} r^{t}$ is such that $r^t=w^q$. Proceeding inductively, the sequence reached in the local priority algorithm under $\beta$ and $\beta'$ at $\succ'$ agree until they reach an allocation $y$ such that $\beta(y)=\{i\}$ and $\beta'(y)\neq \{i\}$. If there is an infeasible allocation later in the sequence under $\beta'$ where $i$ compromises, repeat the step above. Otherwise, $i$ gets a strict improvement. At each step, the sequence $z^{0}\xrightarrow{j_{1}} z^{1}\xrightarrow{j_{2}} \cdots \xrightarrow{j_{p}} z^{p}$  is truncated, so only finitely many steps need to taken. \qed

{\footnotesize
\bibliography{gspbib}
\bibliographystyle{econometrica}
}
\clearpage

\appendix

\end{document}